\documentclass[aps,prx,reprint,superscriptaddress,amsmath,amssymb,floatfix,longbibliography]{revtex4-2}

\usepackage[T1]{fontenc}
\usepackage[utf8]{inputenc}
\usepackage{CJKutf8}
\usepackage{lmodern}
\usepackage{microtype}
\usepackage{mathtools}
\usepackage{amsthm}
\usepackage{bm}
\usepackage{booktabs}
\usepackage{array}
\usepackage{graphicx}
\usepackage{xcolor}
\usepackage{physics}
\usepackage[colorlinks=true,linkcolor=blue!50!black,citecolor=blue!50!black,urlcolor=blue!50!black]{hyperref}

\newcommand{\F}{\mathbb{F}}
\newcommand{\Z}{\mathbb{Z}}

\newcommand{\GSD}{\mathrm{GSD}}
\newcommand{\QtRCC}{\ensuremath{\mathrm{QtRCC}}}
\newcommand{\QtHC}{\ensuremath{\mathrm{QtHC}}}
\newcommand{\ex}{\hat{x}}
\newcommand{\ey}{\hat{y}}
\newcommand{\ez}{\hat{z}}

\newtheorem{theorem}{Theorem}
\newtheorem{lemma}{Lemma}

\begin{document}
\begin{CJK*}{UTF8}{gbsn}
		
\title{Random Local Stabilizer Codes in Three Dimensions without String or Self-Similar Fractal Logical Operators}
		
\author{Han Yan (闫寒)}
\email{hanyan@issp.u-tokyo.ac.jp}
\affiliation{Institute for Solid State Physics, University of Tokyo, Kashiwa, Chiba 277-8581, Japan}

\date{\today}
		
\begin{abstract}
		Quantum error-correcting codes (QECs) are essential components of quantum computation and have deep connections to quantum phases of matter.
		A key obstruction to passive self-correcting QECs is the presence of string logical operators, which can generate logical errors through constant-energy-barrier processes.
		Haah's Codes (fracton codes) showed that  three-dimensional stabilizer codes can forbid such string logical operators, but their translation-invariant structure supports self-similar fractal logical operators with a logarithmic energy barrier.
		We introduce the qutrit random cubic codes, a family of local qutrit Calderbank-Shor-Steane stabilizer Hamiltonians with similar cube-check structure as Haah's Code 1 but built from spatially varying stabilizers.
		We prove that these models retain the no-string property and numerically observe that they have properties distinct from translation-invariant fracton codes: the smallest ground-state degeneracy exponent is $k=2$ for odd $L$ and $k=4$ for even $L$; noncontractible plane-logical operators span the entire logical space; and charge-push diagnostics show that the self-similar fractal operators are absent. 
		These results demonstrate that constrained randomness can fundamentally change the nature of stabilizer codes and improve their self-correction properties. 
		They further point to broader families of quantum error-correcting codes and quantum phases beyond canonical topological and fracton orders.
	\end{abstract}
	
	\maketitle
\end{CJK*}

\section{Introduction}

Self-correcting quantum error-correcting codes (QECs) aim to protect logically encoded quantum information through the equilibrium and dynamical properties of a many-body Hamiltonian, without repeated syndrome extraction and feedback.
They are a key component of fault-tolerant quantum computing.
Stabilizer Hamiltonians have played a central role in constructing QECs~\cite{CalderbankShor1996,Steane1996,Gottesman1997,AlbertFaist2026}.
The toric code established the topological-code paradigm~\cite{Kitaev2003,Dennis2002}, but in two dimensions its logical operators are strings and its energy barrier is constant, making it unsuitable for self-correcting applications.
More generally, two-dimensional local stabilizer memories obey strong no-go and parameter-tradeoff constraints~\cite{BravyiTerhal2009,BravyiPoulinTerhal2010,Terhal2015,Brown2016}.
Three spatial dimensions are therefore the first physically natural arena in which a local stabilizer Hamiltonian might avoid string-driven thermal failure.

Haah's Code 1 (a.k.a. Haah's code, fracton code) was a breakthrough: it is a three-dimensional stabilizer code that can be rigorously proven to have no string logical operators~\cite{Haah2011}.
Its cube-supported Calderbank-Shor-Steane (CSS) checks prevent excitations from being moved by finite-width strings; instead, charges and logical operators are organized by algebraic self-similar fractal recurrences.
This behavior is now understood through energy-barrier scaling, partial self-correction, Laurent-polynomial methods, and the broader theory of fracton order~\cite{Chamon2005,BravyiLeemhuisTerhal2011,BravyiHaah2011,Kim2012,BravyiHaah2013, Haah2013,HaahThesis2013,Yoshida2013,Williamson2016,DuaKimChengWilliamson2019, DuaSarkarWilliamsonCheng2020,VijayHaahFu2015,VijayHaahFu2016, NandkishoreHermele2019,PretkoChenYou2020,AitchisonBulmashDuaDohertyWilliamson2024}, including coupled-layer and foliated descriptions~\cite{MaLakeChenHermele2017,SlagleKim2018,ShirleySlagleWangChen2018, ShirleySlagleChen2019,AasenBulmashPremSlagleWilliamson2020, DevakulWilliamson2021,WilliamsonDevakul2021,SongDuaShirleyWilliamson2023}, tensor gauge theories and multipole conservation laws~\cite{Pretko2017Subdimensional,Pretko2017Generalized,SlagleKim2017, BulmashBarkeshli2018Higgs,BulmashBarkeshli2018FractalDynamics,Gromov2019,LeeHuChoWatanabe2025}, and related condensed-matter settings~\cite{KimHaah2016,PremHaahNandkishore2017,Yan2019HyperbolicFracton, Yan2019HyperbolicFractonII,Yan2020GeodesicString,YanBentonJaubertShannon2020, HeringYanReuther2021,YanSlagleNevidomskyy2022,YanNevidomskyy2024}.
However, in translation-invariant fracton codes such as Haah's Code 1, there exist self-similar fractal operators and associated arithmetic finite-size effects in the ground-state degeneracy, and only logarithmic energy barrier scaling.

Other routes toward more robust three-dimensional QECs use more engineered architectures.
Welded solid codes obtain a polynomial barrier by breaking microscopic translation invariance~\cite{Michnicki2014}; layer-code and defect-network constructions use coupled lower-dimensional systems and defect junctions to approach optimal three-dimensional local-code scaling~\cite{AasenBulmashPremSlagleWilliamson2020,SongDuaShirleyWilliamson2023,WilliamsonBaspin2024}; and product-code, cored-product-code, and recursive constructions provide new routes to fracton models and passive-memory candidates~\cite{TanRobertsTantivasadakarnYoshidaYao2025,RobertsKohTanYao2025,BalasubramanianDavydovaLin2026}.
These examples show that geometry, hierarchy, and nonuniformity can change memory behavior.

Here we explore a different route that keeps the cubic lattice geometry but uses constrained random stabilizers that break translation symmetry.
We study a qutrit generalization of Haah's Code 1 with the same cube-check support, but generalized to qutrits.
Remarkably, the commutation and topology constraints no longer fully determine the cubic stabilizers.
They still impose strong algebraic conditions, but leave enough freedom to vary the stabilizers in space beyond trivial unitary rotations.
The resulting family is the \textit{Qutrit Random Cubic Codes}, denoted \QtRCC.

Our exact analytical result is the same no-string theorem: all cases in \QtRCC{}  have no fixed-width logical string segments defined by Haah.
The proof follows the same strategy of Haah's original proof, which is to deform/erase the support of a candidate string operator, until it can be broken into two disconnected pieces.

We then study \QtRCC{} numerically, using a spatially uniform qutrit Haah's code reference model, denoted \QtHC{}, as the translation-invariant benchmark.
Among sampled models at $10\le L\le25$, the smallest observed ground-state degeneracy exponent is stably $k=2$ for odd $L$ and $k=4$ for even $L$ for \QtRCC{}, in sharp contrast with the arithmetic finite-size dependence of \QtHC{}.
Plane-supported logical operators span the entire logical operator space, while no proper axis-aligned tube logical is detected even at width $L-1$.
Charge-push diagnostics further show that the power-of-three self-similar fractal operator exhibited by \QtHC{} is absent in representative \QtRCC{} instances.
These results suggest that the \QtRCC{} can have improved (partial) self-correcting properties compared to the transnational-symmetric fracton codes. 

The broader implication is that constrained randomness can fundamentally change the behavior of stabilizer QECs.
This points to new families of random stabilizer codes and disorder-modified topological memories, as well as spatially nonuniform fracton and subsystem-symmetry frameworks. 
As a model of condensed matter, \QtRCC{} combines finite-torus ground-state degeneracy with local fracton-like charge dynamics, raising new challenges in understanding and classifying spatially nonuniform stabilizer phases.

The paper is organized as follows.
Section~\ref{sec:model} defines \QtRCC{} and derives the topology and commutation constraints.
Section~\ref{sec:nostringmain} states the no-string theorem and sketches the proof, with its complete version presented in Appendix~\ref{app:nostringproof}.
Section~\ref{sec:numerics} presents ground-state degeneracy and logical-operator diagnostics on finite periodic lattices.
Section~\ref{sec:no.fractal} studies charge-push diagnostics and reports the absence of the   power-of-three fractals in representative sampled instances.
Section~\ref{sec:discussion} summarizes the implications and open problems.

\section{The Qutrit Random Cubic Code}
\label{sec:model}

\subsection{The Check Support Structure}

\begin{figure}[th!]
	\includegraphics[width=\columnwidth]{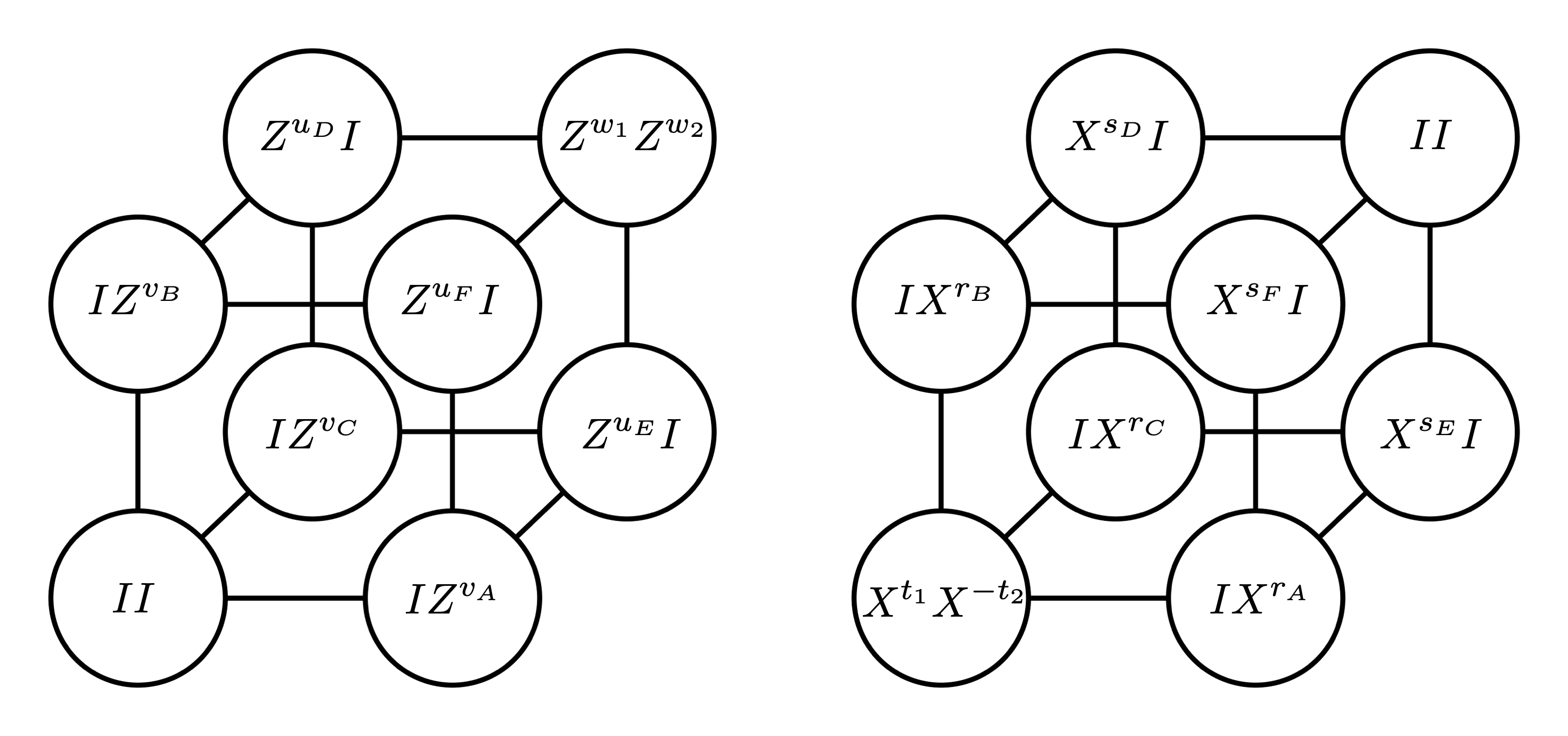}
	\caption{Local cube-check structure for Qutrit Random Cubic Codes (\QtHC).
		The left panel shows the $Z$-check $A_r$ and the right panel shows the reflected $X$-check $B_r$.
		Each cube corner carries two qutrits, and the corresponding operators are displayed.
		With the corner convention of Eq.~\eqref{eq:corners}, the $Z$-check has no support at $O$ and two-qutrit support at $G$, while the $X$-check has two-qutrit support at $O$ and no support at $G$.
		The remaining six corners are one-operator corners.}
	\label{fig:checks}
\end{figure}

We use qutrits and their Pauli operators throughout~\cite{Gottesman1998,AshikhminKnill2001,KetkarKlappeneckerKumarSarvepalli2006}.
A three-level qutrit has computational basis $\{\ket{a}:a\in\F_3\}$, equivalently $a\in\mathbb Z/3\mathbb Z$.
We use the conventional definition
\begin{equation}
	X\ket{a}=\ket{a+1},\qquad Z\ket{a}=\omega^a\ket{a},\qquad
	\omega=e^{2\pi i/3}.
\end{equation}
All additions and Pauli exponents are understood in $\F_3$.
The Pauli operators obey the commutation relation
\begin{equation}
	Z^zX^x=\omega^{zx}X^xZ^z,\qquad x,z\in\F_3 .
\end{equation}
Thus a CSS $Z$ operator with exponent vector $\vb* z$ commutes with a CSS $X$ operator with exponent vector $\vb* x$ if and only if
\begin{equation}
	\vb* z\cdot\vb*  x=\sum_j z_jx_j=0\pmod 3 .
	\label{eq:csscomm}
\end{equation}

For finite-size calculations we use the periodic cubic lattice $\Lambda_L=(\Z/L\Z)^3$, where $\Z$ denotes the integers; equivalently, $\Lambda_L$ is a finite-size lattice with periodic boundary conditions.
Two qutrits are placed at each site $r$, denoted $(r,1)$ and $(r,2)$, so the periodic lattice contains $n=2L^3$ qutrits.
The same local check formulae are also used on the infinite lattice $\Z^3$ when discussing finite-support operators.
We write $\ex,\ey,\ez$ for the positive unit lattice vectors.

For a cube anchored at $r$, we label the eight corners by $r+\mu$, with $\mu\in\{O,A,B,C,D,E,F,G\}$, where
\begin{equation}
	\begin{split}
		&O=(0,0,0),\quad G=(1,1,1), \\
		& A=(1,0,0),\quad B=(0,1,0),\quad C=(0,0,1),\\
		&D=(0,1,1),\quad E=(1,0,1),\quad F=(1,1,0).
	\end{split}
	\label{eq:corners}
\end{equation}
Thus $O$ is the cube anchor and $G=O+\ex+\ey+\ez$ is the opposite body-diagonal corner.
The corners $A,B,C$ are the three positive coordinate neighbors of $O$, and $D,E,F$ are the three face-diagonal corners adjacent to $G$.

We place one $Z$-check $A_r$ and one $X$-check $B_r$ for each cube anchored at $r$.
Figure~\ref{fig:checks} illustrates their local supports and coefficient structure.
Let the nonzero exponent field of the $Z$-check at anchor $r$ be
\begin{equation}
	\mathbf a_r=(w_{1,r},w_{2,r},u_{D,r},u_{E,r},u_{F,r},v_{A,r},v_{B,r},v_{C,r})
	\in(\F_3^*)^8,
\end{equation}
where $\F_3^*=\{1,2\}$.
The $Z$-check is defined as the product of $Z$ operators on the corners
\begin{align}
	A_r={}&Z_{r+G,1}^{w_{1,r}}Z_{r+G,2}^{w_{2,r}}
	Z_{r+D,1}^{u_{D,r}}Z_{r+E,1}^{u_{E,r}}Z_{r+F,1}^{u_{F,r}}
	\nonumber\\
	&\times Z_{r+A,2}^{v_{A,r}}Z_{r+B,2}^{v_{B,r}}Z_{r+C,2}^{v_{C,r}} .
	\label{eq:Zcheck}
\end{align}
The $O$ corner has only identity operators.

Similarly, let
\begin{equation}
	\mathbf b_r=(t_{1,r},t_{2,r},r_{A,r},r_{B,r},r_{C,r},s_{D,r},s_{E,r},s_{F,r})
	\in(\F_3^*)^8
\end{equation}
be the nonzero exponent field of the $X$-check at anchor $r$.
The $X$-check is defined as
\begin{align}
	B_r={}&X_{r+O,1}^{t_{1,r}}X_{r+O,2}^{-t_{2,r}}
	X_{r+A,2}^{r_{A,r}}X_{r+B,2}^{r_{B,r}}X_{r+C,2}^{r_{C,r}}
	\nonumber\\
	&\times X_{r+D,1}^{s_{D,r}}X_{r+E,1}^{s_{E,r}}X_{r+F,1}^{s_{F,r}} .
	\label{eq:Xcheck}
\end{align}
The $G$ corner has only identity operators.
The minus sign in the exponent of $X_{r+O,2}$ is a finite-field sign: $-t_{2,r}=2t_{2,r}\in\F_3$.

For later use, we identify each CSS check on the periodic lattice with its exponent vector in $\F_3^{2L^3}$.
The matrices $H_Z$ and $H_X$ have one row for each cube anchor: $H_Z(r,\cdot)$ is the exponent vector of $A_r$, and $H_X(r,\cdot)$ is the exponent vector of $B_r$.

For qutrit stabilizer checks, replacing a check by its square multiplies all exponents in that check by the nonzero scalar $2\in\F_3$, or equivalently by $-1$.
This leaves the stabilizer group unchanged.
	
The support structure closely resembles the cubic check support pattern used in Haah's Code 1, with qubits replaced by qutrits and with each nonzero exponent (which can only be 1) replaced by a nonzero field element of $\F_3$.
More specifically, the $Z$-check has no support at $O$ and two-qutrit support at $G$, while the $X$-check has the reflected placement: two-qutrit support at $O$ and no support at $G$.
The six remaining corners are one-qutrit corners: $A,B,C$ use the second qutrit and $D,E,F$ use the first qutrit, matching the reflected cubic support structure of Haah's Code 1.

\subsection{Topology and commutation constraints}
\label{sec:constraints}

Equations~\eqref{eq:Zcheck} and \eqref{eq:Xcheck} define a local cube-check structure.
They do not by themselves define a stabilizer code yet, because arbitrary choices of the nonzero coefficients need not commute.
In this subsection we impose two additional requirements: a topology identity for the $Z$-checks and the local CSS commutation equations between $A$ and $B$ checks.
Coefficient fields satisfying these requirements will be called \textit{admissible}.
The random collection of models studied in this work is sampled only from this admissible set.

On the periodic lattice $\Lambda_L$ we impose two constraints.
First, the product of all $Z$-checks is required to be identity.
This gives one global relation among the $Z$-checks and therefore prevents the $2L^3$ local checks from fixing a unique stabilizer state, so that ground-state degeneracy is guaranteed.
Second, every $Z$-check must commute with every $X$-check.
Both conditions are equations for the local qutrit Pauli exponents.

The global $Z$-stabilizer relation is
\begin{equation}
\prod_{r\in\Lambda_L}A_r=I .
\label{eq:topprod}
\end{equation}
Equivalently, the sum of the exponent vectors of all $A_r$ vanishes in $\F_3^{2L^3}$.
Hence the rows of $H_Z$ obey at least one linear relation, so $\rank(H_Z)\le L^3-1$.
For the resulting commuting CSS stabilizer group, there are $2L^3$ physical qutrits and at most $L^3$ independent $X$-checks, so this imposed $Z$-check dependence guarantees
\begin{equation}
k\equiv \log_3\GSD=2L^3-\rank(H_Z)-\rank(H_X)\ge 1 .
\end{equation}
Here $\GSD$ denotes the ground-state degeneracy.
We do not impose a corresponding global relation on the $X$-checks; one stabilizer dependence is sufficient to obtain a nonzero $k$.

The identity \eqref{eq:topprod} holds if and only if the total $Z$ exponent contributed by different $Z$-checks on each physical qutrit vanishes.
In the indexing convention of Eq.~\eqref{eq:Zcheck}, the first qutrit receives contributions from one $G$-corner coefficient and from the $D,E,F$ corners of three neighboring cubes; the second qutrit receives contributions from one $G$-corner coefficient and from the $A,B,C$ corners of three neighboring cubes.
This gives two equations at every site:
\begin{equation}
\begin{aligned}
w_{1,r}+u_{D,r+\ex}+u_{E,r+\ey}+u_{F,r+\ez}= 0,\\
w_{2,r}+v_{A,r+\ey+\ez}
     +v_{B,r+\ex+\ez}
     +v_{C,r+\ex+\ey}= 0.
\end{aligned}
\label{eqn:topo.constraint}
\end{equation}
These are finite-field equations with all coefficient variables restricted to the nonzero set $\F_3^*=\{1,2\}$.

We next impose local CSS commutation.
By Eq.~\eqref{eq:csscomm}, every overlapping pair of a $Z$-check and an $X$-check must commute.
Fix an $X$-check $B_r$.
Among the $27$ neighboring $Z$-checks, $14$ produce nontrivial constraints on the coefficient fields.
The raw finite-field equations are listed in Appendix~\ref{app:commutation}.

Replacing a stabilizer check by its square rescales all its exponents by the nonzero scalar $2\in\F_3$ and leaves the stabilizer group unchanged.
We use this check normalization by choosing $t_{1,r}\in\F_3^*$, after which seven of the local commutation equations fix the remaining $X$-check exponents in terms of nearby $Z$-check exponents.

Define the local ratios
\begin{equation}
\begin{aligned}
\alpha_D(r)&=\frac{u_{D,r+\ey+\ez}}{w_{2,r+\ey+\ez}},
&
\alpha_E(r)&=\frac{u_{E,r+\ex+\ez}}{w_{2,r+\ex+\ez}},
\\
\alpha_F(r)&=\frac{u_{F,r+\ex+\ey}}{w_{2,r+\ex+\ey}},
&
\beta_A(r)&=\frac{v_{A,r+\ex}}{w_{1,r+\ex}},
\\
\beta_B(r)&=\frac{v_{B,r+\ey}}{w_{1,r+\ey}},
&
\beta_C(r)&=\frac{v_{C,r+\ez}}{w_{1,r+\ez}},
\\
\rho(r)&=\frac{w_{1,r+\ex+\ey+\ez}}{w_{2,r+\ex+\ey+\ez}} .
\end{aligned}
\label{eq:alpha_beta_rho}
\end{equation}
The seven anchoring commutation equations give
\begin{equation}
\begin{aligned}
r_{A,r}&=-\alpha_D(r)t_{1,r},
&
r_{B,r}&=-\alpha_E(r)t_{1,r},
\\
r_{C,r}&=-\alpha_F(r)t_{1,r},
&
s_{D,r}&=\rho(r)\beta_A(r)t_{1,r},\\
s_{E,r}&=\rho(r)\beta_B(r)t_{1,r},
&
s_{F,r}&=\rho(r)\beta_C(r)t_{1,r},\\
t_{2,r}&=\rho(r)t_{1,r}.
\end{aligned}
\label{eq:derive2}
\end{equation}
Thus, once the $Z$-check coefficient field and the local normalization $t_{1,r}$ are fixed, all exponents in $X$-checks $B_r$ are fixed by local commutation.

There are seven more nontrivial $A$-$B$ commutation equations.
Substituting Eq.~\eqref{eq:derive2} into them gives seven constraints on the $Z$-check coefficient field alone:
\begin{equation}
\begin{aligned}
 &u_{D,r-\ex+\ez}\rho(r)\beta_C(r)
    -v_{C,r-\ex+\ez}\alpha_D(r)=0,\\
 &u_{D,r-\ex+\ey}\rho(r)\beta_B(r)
    -v_{B,r-\ex+\ey}\alpha_D(r)=0,\\
 &u_{E,r-\ey+\ez}\rho(r)\beta_C(r)
    -v_{C,r-\ey+\ez}\alpha_E(r)=0,\\
 &u_{F,r+\ey-\ez}\rho(r)\beta_B(r)
    -v_{B,r+\ey-\ez}\alpha_F(r)=0,\\
 &u_{E,r+\ex-\ey}\rho(r)\beta_A(r)
    -v_{A,r+\ex-\ey}\alpha_E(r)=0,\\
 &u_{F,r+\ex-\ez}\rho(r)\beta_A(r)
    -v_{A,r+\ex-\ez}\alpha_F(r)=0,\\
 &u_{D,r}\rho(r)\beta_A(r)
    +u_{E,r}\rho(r)\beta_B(r)
    +u_{F,r}\rho(r)\beta_C(r)\\
&\quad
    -v_{A,r}\alpha_D(r)
    -v_{B,r}\alpha_E(r)
    -v_{C,r}\alpha_F(r)=0.
\end{aligned}
\label{eqn:com.constraint}
\end{equation}
The first six equations are two-term edge-overlap constraints.
The final equation is the same-cube commutation condition between $A_r$ and $B_r$.
As before, all coefficient variables appearing in these nonlinear equations are required to be nonzero.

On the periodic lattice we call a $Z$-coefficient field \textit{admissible} if all coefficients are nonzero and Eqs.~\eqref{eqn:topo.constraint} and \eqref{eqn:com.constraint} hold at every site.
Given such a field and a choice of nonzero check normalizations $t_{1,r}$, Eqs.~\eqref{eq:alpha_beta_rho} and \eqref{eq:derive2} define nonzero $X$-check coefficients, and the resulting CSS stabilizer checks commute.
On $\Z^3$ there is no finite global product relation corresponding to Eq.~\eqref{eq:topprod}.
Nevertheless, we impose the local coefficient equations \eqref{eqn:topo.constraint} as part of the formal definition of a locally admissible coefficient field.
For formal infinite-lattice statements, admissible means that all coefficients are nonzero and Eqs.~\eqref{eqn:topo.constraint} and \eqref{eqn:com.constraint} hold at every site of $\Z^3$.

For an admissible coefficient field on $\Lambda_L$, the checks commute and define the qutrit stabilizer Hamiltonian
\begin{equation}
H_{\QtRCC}=-\sum_{r\in\Lambda_L}\left(A_r+A_r^\dagger+B_r+B_r^\dagger\right).
\label{eq:hamiltonian}
\end{equation}
The ground space is the common $+1$ eigenspace of all $A_r$ and $B_r$.
Equivalently, for any order-three Pauli stabilizer $S$, the projector onto its $+1$ eigenspace is
\begin{equation}
\Pi(S)=\frac{I+S+S^2}{3},
\end{equation}
and $S+S^\dagger=3\Pi(S)-I$.
Thus Eq.~\eqref{eq:hamiltonian} has the same ground space as $-\sum_S\Pi(S)$, or equivalently as the positive penalty Hamiltonian $\sum_S[I-\Pi(S)]$, up to an additive constant and an overall positive scale.

For the corresponding qubit model, this coefficient freedom is absent: every nonzero Pauli exponent is equal to one in $\F_2$.
This recovers Haah's Code 1.
Remarkably, for qutrit models, the two nonzero exponents $1$ and $2$ allow spatially varying $X$-check and $Z$-check coefficients while satisfying the topology and commutation equations, and allow one to construct many unitarily inequivalent models.
This is the key coefficient freedom underlying the \QtRCC{} construction.

\subsection{The random and uniform models}
\label{sec:qtrccqthc}

\begin{figure}[th!]
``\includegraphics[width=\columnwidth]{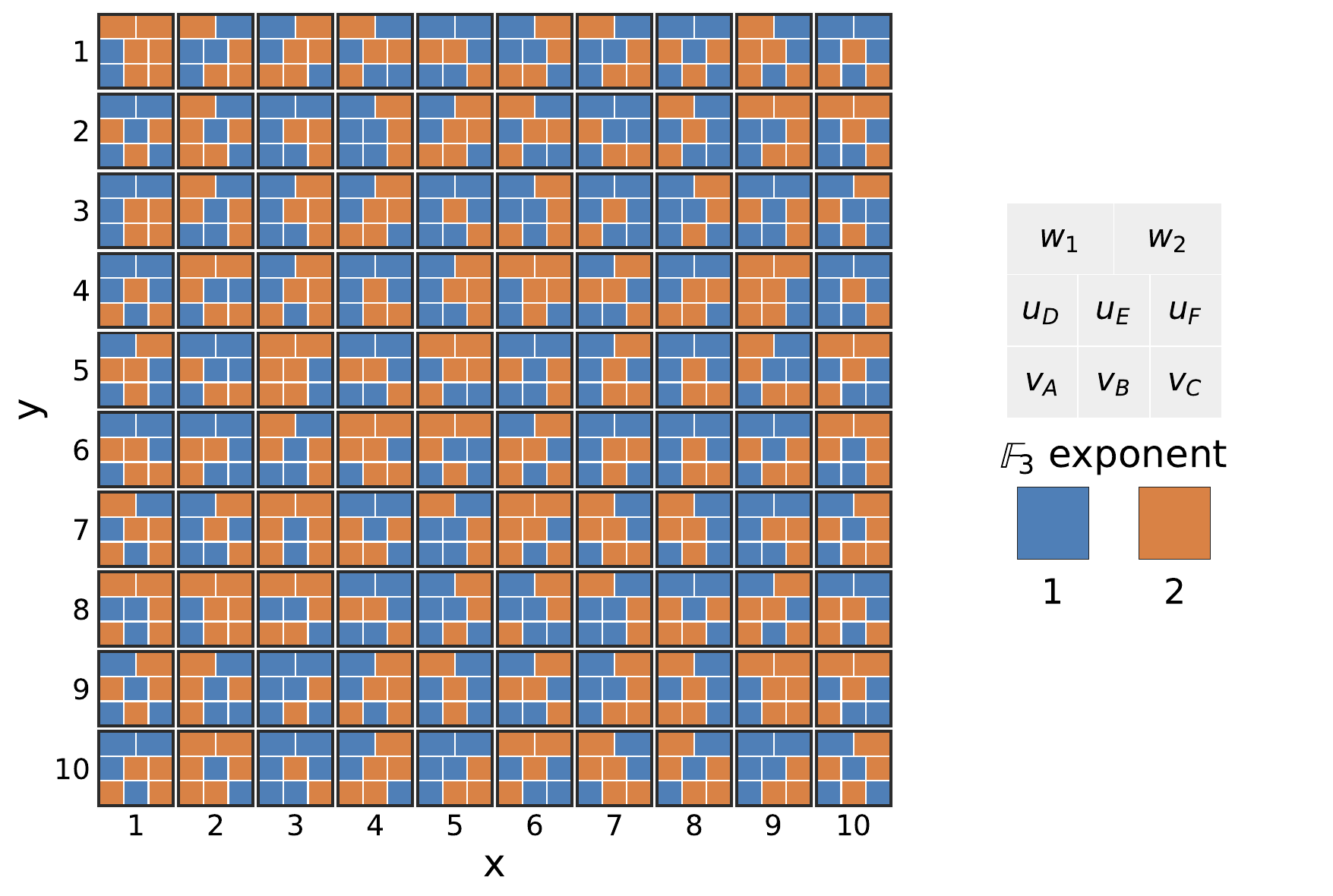}
\caption{Example of an admissible $Z$-check coefficient field.
The figure shows the $z=0$ slice of the curated $L=10$ example \textsf{L10\_model\_01}~\cite{QTRCCRepository}.
Each square is one cube anchor and contains a compact glyph for the eight nonzero coefficients $(w_1,w_2,u_D,u_E,u_F,v_A,v_B,v_C)$ of the local $Z$-check.
Blue denotes the finite-field value $1\in\F_3$, and orange denotes $2\in\F_3$.
With translation invariance along the $[111]$ direction, this plot defines the stabilizer code unambiguously.}
\label{fig:L10zslice}
\end{figure}

The preceding subsection defines the admissible coefficient fields.
A finite-lattice \QtRCC{} example is obtained by choosing a nonzero $Z$-check coefficient field $\{\mathbf a_r\}_{r\in\Lambda_L}$ satisfying the topology equations \eqref{eqn:topo.constraint} and the reduced commutation constraints \eqref{eqn:com.constraint}.
After fixing the conventional normalization of each $X$-check, the coefficients of $B_r$ are then determined locally by Eqs.~\eqref{eq:alpha_beta_rho} and \eqref{eq:derive2}.
Thus the random degrees of freedom in \QtRCC{} are admissible finite-field coefficient fields; arbitrary independent choices of the $X$ and $Z$ coefficient fields do not in general define a commuting stabilizer Hamiltonian.
Throughout this paper, \QtRCC{} refers to the full locally admissible constrained-random family. 

In the numerical work, fully three-dimensional generation of random fields under these nonlinear constraints remains computationally challenging.
We therefore use a line-symmetric subcollection of periodic-lattice admissible fields, which is tractable at the lattice sizes studied here.
Specifically, the $Z$-check coefficient field is taken to be invariant under translation along the $[111]$ direction,
\begin{equation}
\mathbf a_{r+\ell(\ex+\ey+\ez)}=\mathbf a_r,
\qquad \ell\in\Z/L\Z .
\label{eq:111}
\end{equation}
Equivalently, the eight coefficients $w_{1,r}$, $w_{2,r}$, $u_{D,r}$, $u_{E,r}$, $u_{F,r}$, $v_{A,r}$, $v_{B,r}$, and $v_{C,r}$ are constant along each $[111]$ line.
The derived $X$-check coefficient field has the same line symmetry.
This restriction allows spatial variation in the two transverse directions.
Figure~\ref{fig:L10zslice} shows one $z=0$ slice of the $Z$-check coefficient field for an $L=10$ example~\cite{QTRCCRepository}.

We also use a spatially uniform qutrit cubic reference model, denoted \QtHC{}.
It is the qutrit analogue of Haah's Code 1, with the same local cube support as \QtRCC{} and with all finite-field coefficients independent of the anchor $r$. 
One convenient uniform solution for the $Z$-check coefficients is
\begin{equation}
\begin{aligned}
w_1&=w_2=u_D=v_A=1,\\
u_E&=u_F=v_B=v_C=2 .
\end{aligned}
\label{eq:qthcA}
\end{equation}
The corresponding uniform $X$-check coefficients are
\begin{equation}
\begin{aligned}
t_1&=t_2=r_B=r_C=s_D=1,\\
r_A&=s_E=s_F=2 .
\end{aligned}
\label{eq:qthcB}
\end{equation}
This spatially uniform qutrit model serves as the translation-invariant cubic reference for comparisons with the constrained-random \QtRCC{} collection of models.

\section{Absence of string logical operators in QtRCC}
\label{sec:nostringmain}

\begin{figure}[th!]
\includegraphics[width=\columnwidth]{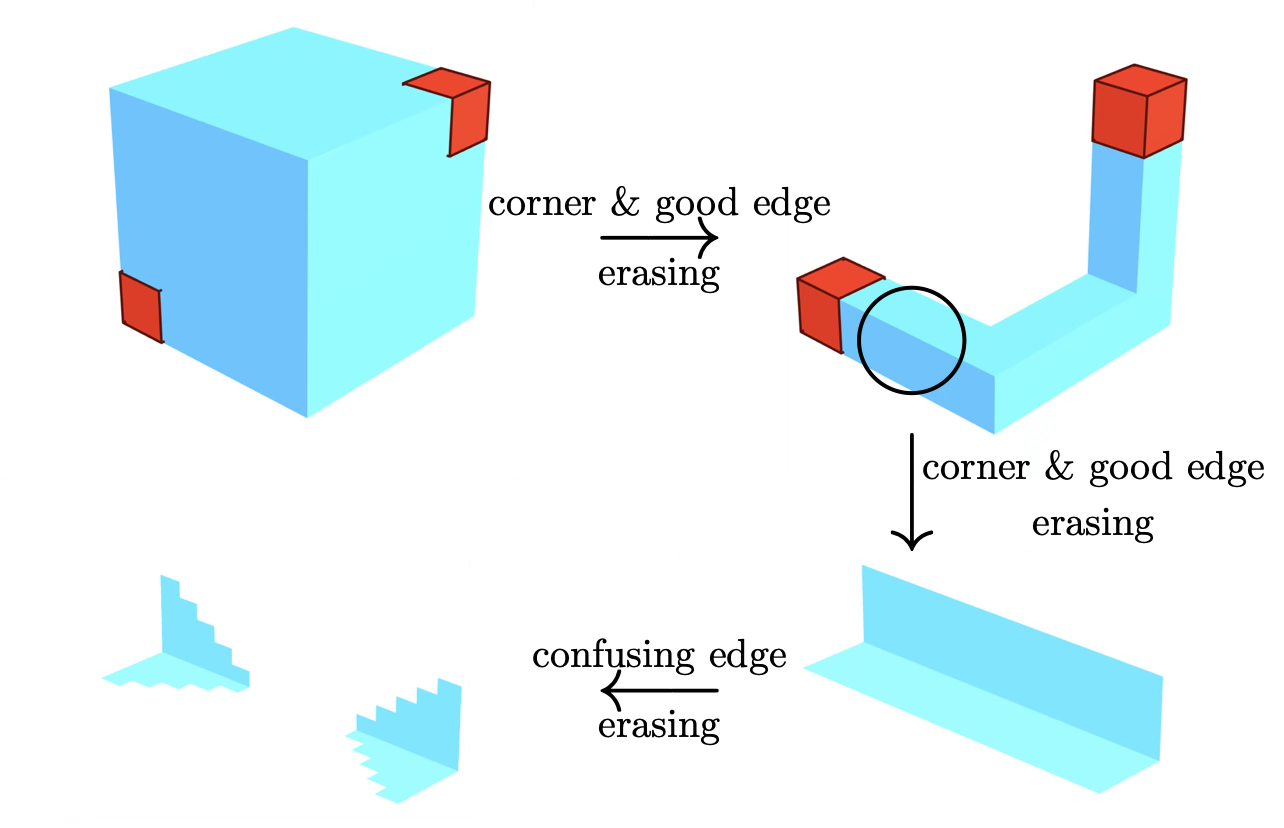}
\caption{Schematic of the no-string proof strategy.
A large connected support has stabilizer-violation regions shown as red cubes.
Corner and good-edge erasing deform the representative to an axis-aligned corridor connecting the same endpoint charges.
Further erasing gives the two flat rectangles.
The confusing constraint then removes rows from the exposed edges: the white erased interval is longest on the outer row and shortens row by row as the induction moves inward, eventually cutting the flat segment open if it is sufficiently long.
The disconnected representative contradicts nontriviality of a sufficiently long fixed-width logical string segment.}
\label{fig:nostringproof}
\end{figure}

This section states the no-string theorem and summarizes why Haah's Code 1 proof extends to \QtRCC{}~\cite{Haah2011}.
The proof outline is essentially the same as  Haah's original proof.
The new step is to verify that each local algebraic ingredient of that proof remains valid when the binary coefficients are generalized to nonzero, spatially varying elements of $\F_3$.
The detailed full proof is given in Appendix~\ref{app:nostringproof}.

We use Haah's definition of a logical string segment~\cite{Haah2011}.
A width-$w$ logical string segment is a finite Pauli operator $P$ together with two congruent width-$w$ anchor cubes $\Omega_1$ and $\Omega_2$.
All stabilizer checks away from the anchors commute with $P$, so any violated checks are confined to the two anchors.
Two such representatives are equivalent if they differ by multiplication by finitely many stabilizer checks.
The segment is nontrivial if every equivalent representative remains connected between the two anchors, and trivial if multiplication by a finite collection of stabilizer checks makes it disconnected.
If $\phi(w)$ is the maximum length of a nontrivial width-$w$ segment, then absence of string logical operators means $\phi(w)<\infty$ for every fixed finite width $w$.
Equivalently, when the two anchors are sufficiently far apart compared with their width, stabilizer multiplication can always transform the segment into a disconnected representative.

\begin{theorem}[No fixed-width logical strings]
\label{thm:nostring}
Consider a locally admissible \QtRCC{} coefficient field on $\Z^3$.
Then there is a finite function $\phi(w)$ such that every nontrivial logical string segment of width $w$ has length at most $\phi(w)$.
With the effective-width convention $W=w+4$, the proof gives the conservative bound
\begin{equation}
\phi(w)\le 27W+48=27(w+4)+48 .
\label{eq:nostring-bound}
\end{equation}
The same statement holds on the periodic lattice $\Lambda_L$ for string segments whose support and anchor collar are contained in a contractible no-wraparound region.
It does not forbid noncontractible plane-supported logical operators on the periodic lattice.
\end{theorem}

The proof is local and deterministic.
It uses the cube support structure, the reflected CSS check pattern, local $A$--$B$ commutation, and the fact that every coefficient is nonzero in $\F_3$.
It does not use the topology identity, the $[111]$ line symmetry, or any statistical property of the sampling procedure.

Because the code is CSS, a Pauli operator may be written, up to phase, as $P=P_XP_Z$.
The $X$-type component is constrained by commutation with the $Z$ checks, while the $Z$-type component is constrained by commutation with the $X$ checks.
It is therefore enough to describe the local argument for an $X$-type component; the $Z$-type argument is obtained by exchanging $X$ and $Z$ and using the reflected check pattern.

The proof has three stages, shown schematically in Fig.~\ref{fig:nostringproof} and sketched below.  
First, corner erasing removes exposed cube corners without creating support outside the prescribed string region.
In other words, the support on an exposed corner must commute with the one-corner overlapping $Z$-check, so this corner is either trivially $II$, or a non-trival operator but generated by the corner of the $X-$check geometrically contained in $P$.
For \QtRCC{}, the only local cube corner that cannot serve as an erasing corner is an identity corner.
Every nonidentity one-qutrit corner imposes an invertible one-component constraint, and every mixed two-qutrit corner imposes a one-dimensional finite-field kernel spanned by the corresponding reflected check vector.
Second, good-edge erasing removes exposed edge support.
For a good edge, the opposite Pauli-type check imposes one nontrivial two-site linear constraint; its kernel is spanned by two endpoint directions removed by corner erasing and one correlated edge direction supplied by the corresponding two-term $A$--$B$ commutation equation.
These two local erasing mechanisms reproduce Haah's support-contained deformation to at most three coordinate-parallel flat corridors and then to the two-rectangle normal form.

The final stage begins after the support has been flattened.
At this point the relevant local condition is Haah's exposed-edge confusing constraint.
Along the outer edge of a thickness-one rectangle, consecutive local commutation equations form a triangular finite-field recurrence: one equation kills a second-qutrit exponent, and the shifted two-site equation then kills the first-qutrit exponent one site farther along the edge.
Because all coefficients are nonzero, the recurrence forces the exposed edge to become identity after a bounded transient.
Repeating the argument row by row creates an empty middle interval in any sufficiently long fixed-width flat segment.
The resulting representative is disconnected, contradicting nontriviality of the assumed logical string segment.
This proves Theorem~\ref{thm:nostring}; the detailed proof is given in Appendix~\ref{app:nostringproof}.

\section{Ground-state degeneracy and membrane logical operators}
\label{sec:numerics}

We now turn to numerical results for finite-size \QtRCC{} models with periodic boundary conditions, to compare with the uniform \QtHC{} reference and to reveal their new properties.

As discussed in Sec.~\ref{sec:qtrccqthc}, the sampled examples used here are admissible coefficient fields with translation invariance along the $[1,1,1]$ direction.
This section studies the ground-state degeneracy and observed logical-operator geometry.
The next section studies charge-push diagnostics for the self-similar recurrence of the uniform reference model.

\subsection{Ground-state degeneracy}
\label{sec:gsd}

\begin{figure}[th!]
\includegraphics[width=\columnwidth]{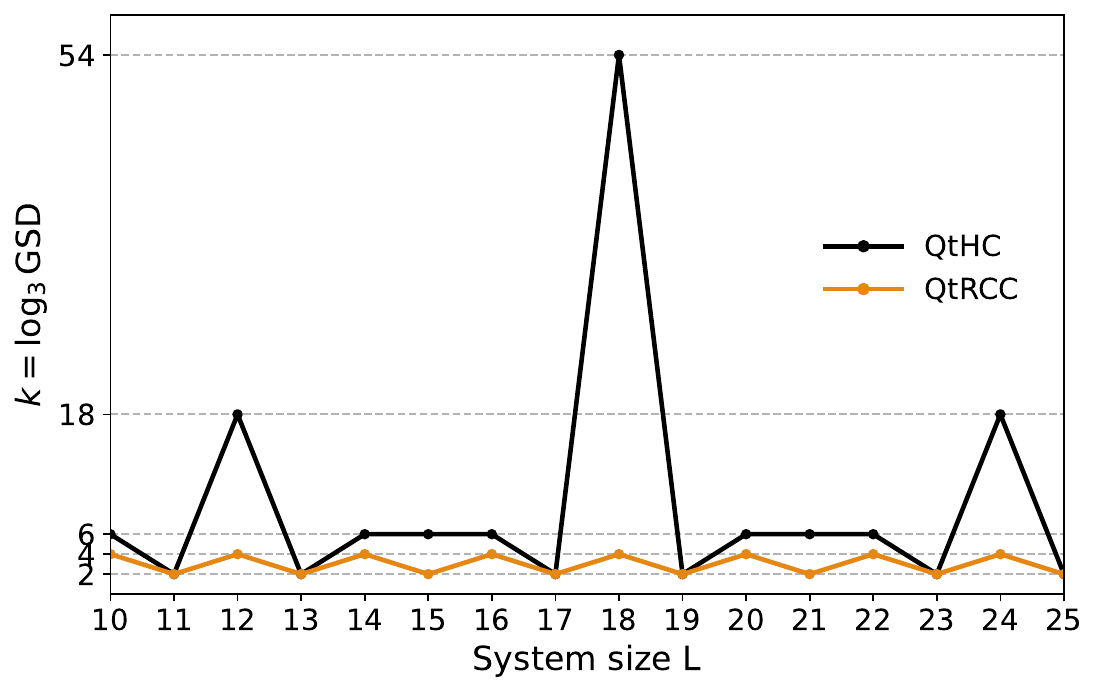}
\caption{Ground-state degeneracy $k = \log_3\GSD $ comparison for $10\le L\le25$.
The translation-invariant \QtHC{} reference model exhibits arithmetic finite-size fluctuations, including the large spike $k=54$ at $L=18$.
The retained minimal-degeneracy \QtRCC{} collection follows the parity pattern of Eq.~\eqref{eq:qtrcck}, with $k=2$ for odd $L$ and $k=4$ for even $L$.}
\label{fig:k.dependency}
\end{figure}

We first compute
\begin{equation}
k=\log_3\GSD
\end{equation}
for generated admissible \QtRCC{} models at fixed system size $L$.
Different random coefficient fields at the same $L$ can have different ground-state degeneracies.
Across the generated models, the smallest observed value follows the simple parity pattern
\begin{equation}
k=\begin{cases}
2,& L\ \text{odd},\\
4,& L\ \text{even}.
\end{cases}
\label{eq:qtrcck}
\end{equation}
for every $10\le L\le25$ studied here.

For the remaining diagnostics we keep, at each $L$, only those generated models whose degeneracy reaches this minimum.
We refer to this subset as the \textit{curated} minimal-degeneracy collection~\cite{QTRCCRepository}.
The number of retained models at each system size is listed in Table~\ref{tab:qtrcc.results}.
Every retained model satisfies
\begin{equation}
\rank(H_Z)=\rank(H_X)=L^3-k/2,
\end{equation}
indicating that the observed degeneracy originates from equal rank deficiencies of the two CSS check matrices.

The contrast with the \QtHC{} reference, shown in Table~\ref{tab:qtrcc.results} and Fig.~\ref{fig:k.dependency}, is pronounced.
The translation-invariant \QtHC{} exhibits strong arithmetic size dependence, including large degeneracy spikes such as $k=54$ at $L=18$.
In contrast, the retained minimal-degeneracy \QtRCC{} examples have the parity pattern of Eq.~\eqref{eq:qtrcck}.

The data show no arithmetic degeneracy spikes of the kind seen in the uniform \QtHC{} reference over the studied sizes.
This indicates that the logical operator structure is highly sensitive to spatial variations of the check coefficients.
Within the curated collection, constrained randomness suppresses the finite-size self-similar fractal mechanism responsible for the enhanced degeneracies of \QtHC{} and leaves only the observed parity dependence of Eq.~\eqref{eq:qtrcck}.

In this sense, the \QtRCC{} GSD are topological instead of geometrical/fractal. This is distinct from type I/II fracton orders, where GSD has non-trivial dependence on the system size. 
The randomness of stabilizers, although locally are of fracton order type, turns the fracton order into a topologically order.

\begin{table*}[t]
\caption{Periodic-lattice ground-state degeneracies of the uniform reference model (\QtHC{}) and the curated minimal-degeneracy \QtRCC{} collection~\cite{QTRCCRepository}.
Among the generated line-symmetric \QtRCC{} examples, the retained samples are those with the smallest observed $k$ at each fixed $L$, namely $k=2$ for odd $L$ and $k=4$ for even $L$.
The final two columns report the minimum found plane-supported quotient-logical representative weights in the $X$- and $Z$-type calculations after minimizing over all curated models at fixed $L$.
These weights come from the plane-logical quotient calculation and should not be read as global minimal code distances.}
\label{tab:qtrcc.results}
\begin{ruledtabular}
\begin{tabular}{c c c c c c}
$L$ & $k_{\QtHC}$ & \QtRCC{} models & $k_{\QtRCC}$ &
$w_X^{\rm min}$ & $w_Z^{\rm min}$ \\
\hline
10 & 6 & 50 & 4 & 50 & 50 \\
11 & 2 & 50 & 2 & 121 & 121 \\
12 & 18 & 50 & 4 & 72 & 72 \\
13 & 2 & 50 & 2 & 169 & 169 \\
14 & 6 & 50 & 4 & 98 & 98 \\
15 & 6 & 20 & 2 & 225 & 225 \\
16 & 6 & 20 & 4 & 128 & 128 \\
17 & 2 & 20 & 2 & 289 & 289 \\
18 & 54 & 20 & 4 & 162 & 162 \\
19 & 2 & 20 & 2 & 361 & 361 \\
20 & 6 & 5 & 4 & 200 & 200 \\
21 & 6 & 3 & 2 & 441 & 441 \\
22 & 6 & 5 & 4 & 242 & 242 \\
23 & 2 & 5 & 2 & 529 & 529 \\
24 & 18 & 1 & 4 & 288 & 288 \\
25 & 2 & 5 & 2 & 625 & 625 \\
\end{tabular}
\end{ruledtabular}
\end{table*}

\subsection{The logical operators}
\label{sec:planes}

\begin{figure}[th!]
\includegraphics[width=\columnwidth]{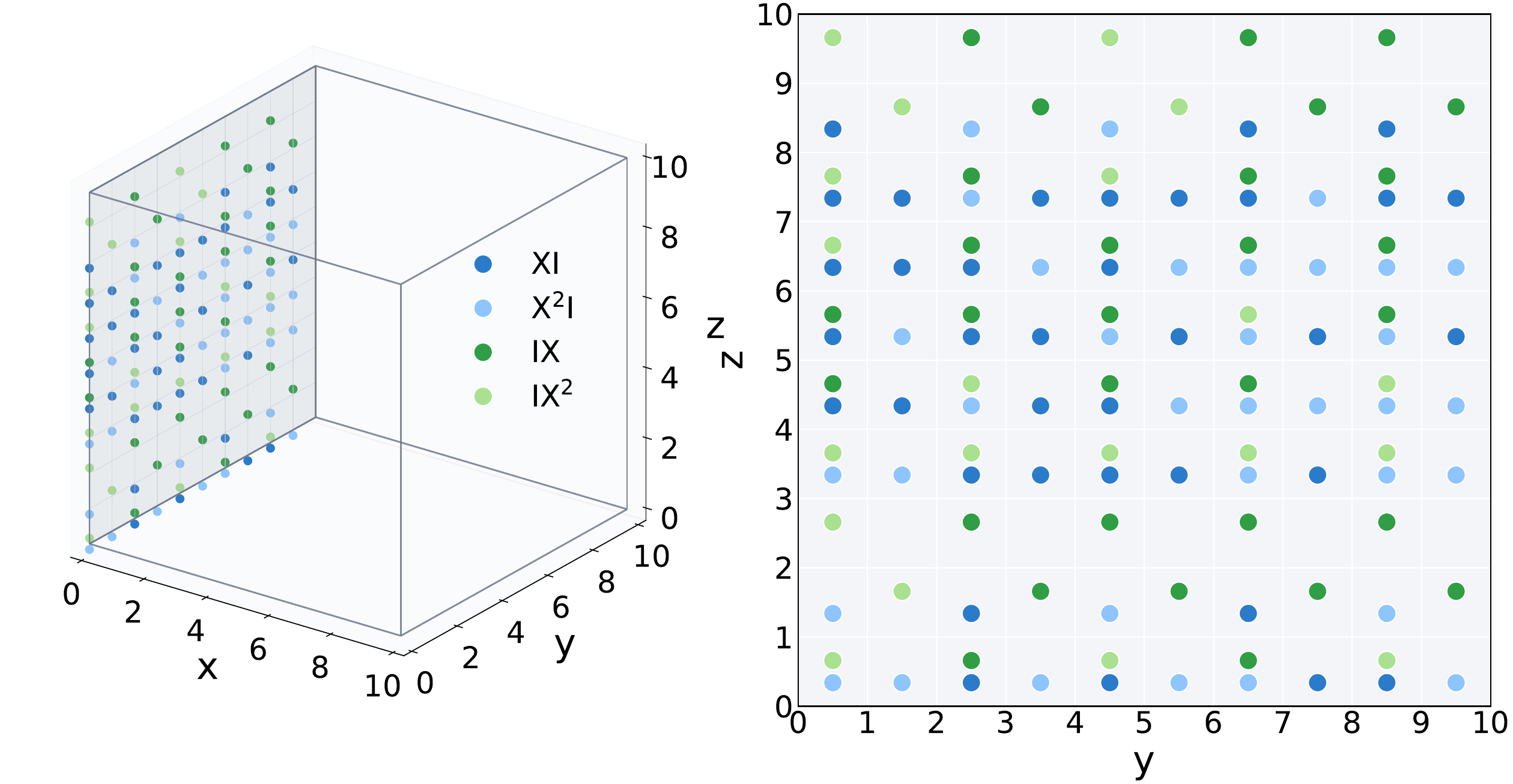}
\caption{Example plane-supported $X$ logical representative for the curated \QtRCC{} example \textsf{L10\_model\_01}~\cite{QTRCCRepository}.
The left panel gives a three-dimensional view of the operator support on the plane $x=0$, and the right panel shows the same support projected to the transverse $y$-$z$ plane.
Blue markers denote first-qutrit factors $XI$, green markers denote second-qutrit factors $IX$, and the two shades distinguish exponents $1$ and $2$ in $\F_3$.
This selected quotient-basis representative has weight $120$: it occupies $90$ of the $10^2$ lattice sites in the plane, with $30$ sites carrying both qutrit factors.}
\label{fig:planeL10}
\end{figure}

\begin{figure}[th!]
\includegraphics[width=\columnwidth]{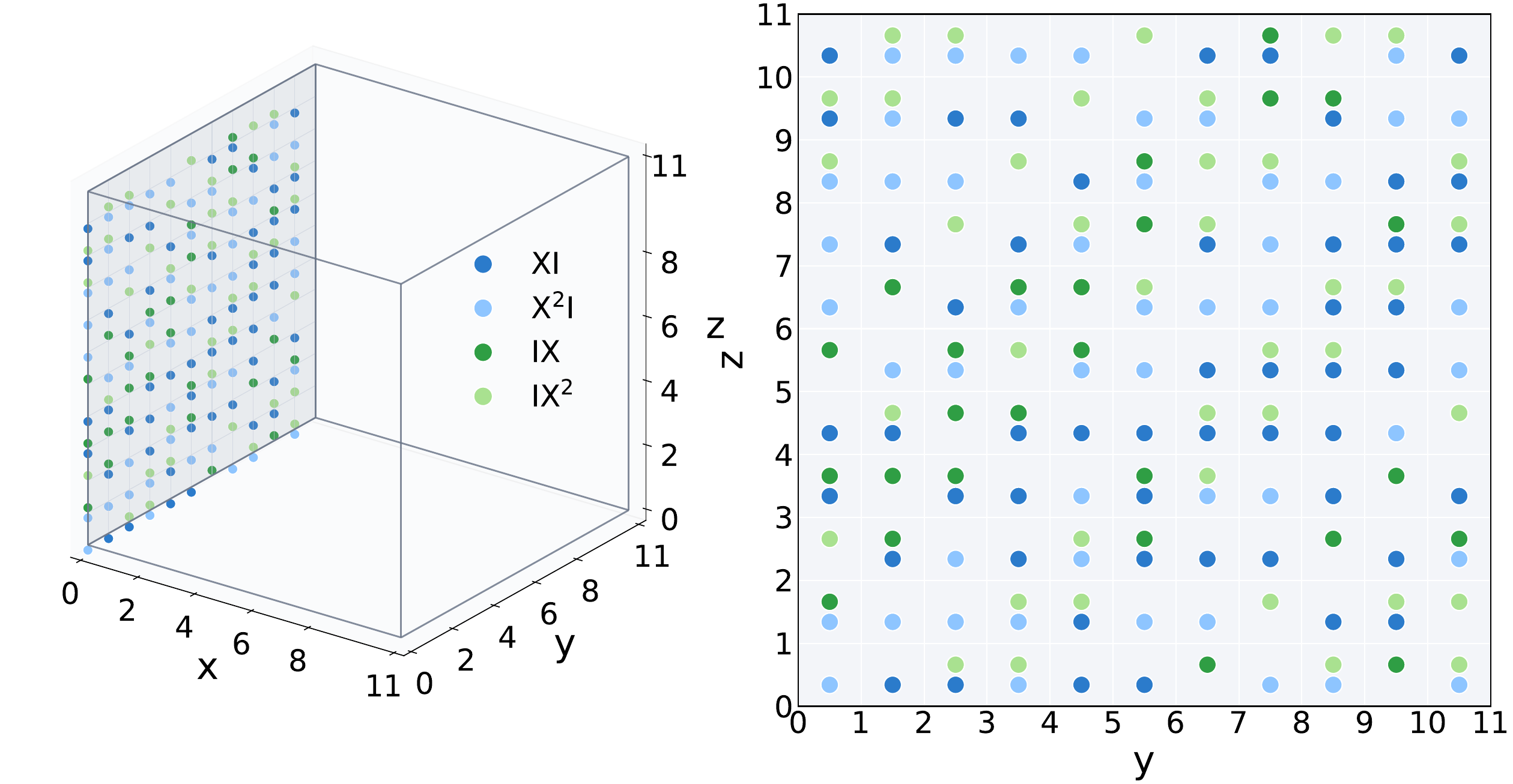}
\caption{Example plane-supported $X$ logical representative for the curated \QtRCC{} example \textsf{L11\_model\_01}~\cite{QTRCCRepository}.
The left panel gives a three-dimensional view of the operator support on the plane $x=0$, and the right panel shows the same support in transverse $y$-$z$ coordinates, using the color convention of Fig.~\ref{fig:planeL10}.
This model has $k=2$, so the $X$-type plane quotient has one selected basis representative on the displayed plane.
The representative has weight $165$ and occupies every site of the $11^2$ plane at least once; $77$ sites carry one nonzero qutrit factor and $44$ sites carry both.}
\label{fig:planeL11}
\end{figure}

\begin{figure}[th!]
\includegraphics[width=\columnwidth]{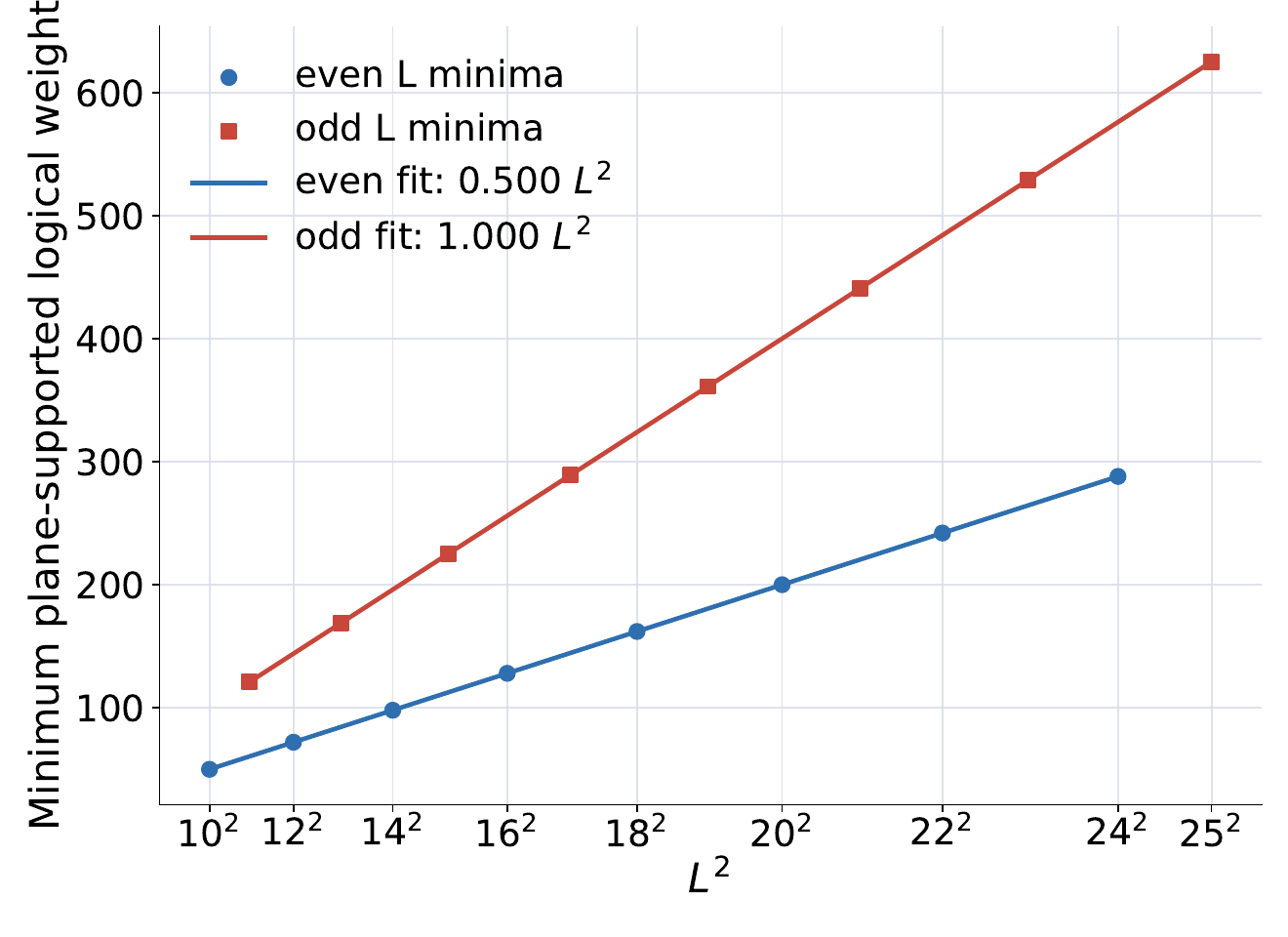}
\caption{Minimum plane-supported logical operator weight versus $L^2$ for the curated \QtRCC{} collection.
For each $L=10,\ldots,25$, the point is the minimal weight over all curated models at that size.
The $X$- and $Z$-type minima coincide at every size.
The best parity-resolved fit is exact on the available data: $w=L^2/2$ for even $L$ and $w=L^2$ for odd $L$.}
\label{fig:planeweight}
\end{figure}

The second numerical observation concerns the geometry associated with the logical operators.
For every curated \QtRCC{} example, we find that noncontractible plane-supported logical operators span the entire logical operator space.

Specifically, we use an exhaustive diagnostic within a precisely specified support family: single-layer coordinate planes normal to the $x$, $y$, and $z$ directions, with every plane position scanned for every curated \QtRCC{} example.

For a plane support $S$, an $X$-type Pauli operator supported on $S$ is charge-neutral precisely when its exponent vector lies in the restricted kernel $\ker H_Z[:,S]$.
Let $P_X(S)$ be a matrix whose rows span these plane-supported charge-neutral candidates.
The logical content carried by this plane, modulo the subgroup generated by $X$-checks, is measured by the quotient rank
\begin{equation}
k_X(S)=
\rank\begin{pmatrix}
H_X\\
P_X(S)
\end{pmatrix}
-\rank(H_X).
\label{eq:planeX}
\end{equation}
The corresponding $Z$-type diagnostic exchanges $X$ and $Z$:
\begin{equation}
k_Z(S)=
\rank\begin{pmatrix}
H_Z\\
P_Z(S)
\end{pmatrix}
-\rank(H_Z),
\label{eq:planeZ}
\end{equation}
where rows of $P_Z(S)$ span the $Z$-type plane-supported candidates in $\ker H_X[:,S]$.

After collecting candidates from all scanned coordinate planes, we find for all 374 curated \QtRCC{} samples that
\begin{equation}
k_X=k_Z=k .
\label{eq:plane-span-full}
\end{equation}
Here $k_X$ and $k_Z$ denote the dimension of logical operators in the $X$ and $Z$ sector respectively. 

Figures~\ref{fig:planeL10} and \ref{fig:planeL11} show two concrete representatives from this plane-supported logical operator calculation.
Both examples are $X$-type representatives on the coordinate plane $x=0$.
In each figure, the left panel shows the support as a three-dimensional plane embedded in the torus, while the right panel shows the same operator in  2D plots  of $y$-$z$ coordinates.
These examples show what the plane-supported logical operators look like microscopically: the support is confined to one noncontractible plane, but the two qutrits at a site and the two nonzero $\F_3$ exponents are populated in a spatially structured pattern.

We next aggregate the plane-logical operator weight data over the same curated collection of models.
For each size $L$ and each CSS Pauli type, we take the overall minimum of the per-model minimum plane-supported logical operator weight over all curated models at that size.
The resulting $X$- and $Z$-type minima are identical for every $10\le L\le25$, as listed in Table~\ref{tab:qtrcc.results} and shown in Fig.~\ref{fig:planeweight}.
Moreover, the sequence is exactly described by a parity-resolved area law over the studied range:
\begin{equation}
w_X^{\rm min}(L)=w_Z^{\rm min}(L)=
\begin{cases}
L^2/2,& L\ \text{even},\\
L^2,& L\ \text{odd}.
\end{cases}
\label{eq:planeweight-parity}
\end{equation} 

We note that this statement is not an exact code-distance result.
Finding a globally minimum-weight logical operator is a global code-distance calculation that we do not perform here.
The searches reported here do not exclude irregular non-plane logical representatives of smaller support, although we have not found any.

Finally a remark is due. 
In the original qubit Haah Code 1, Haah observed that a particular simple type of plane operator, obtained by repeating the same single-site operator over a plane, does not span the entire logical space for some system sizes~\cite{Haah2011}.
Our plane diagnostic is different: for each single coordinate plane, we include all Pauli operators supported on that plane that commute with the stabilizer checks, and then quotient by stabilizers.
With this exhaustive plane-supported search, we find that single-plane operators span the full logical space over the sizes studied, both for the uniform \QtHC{} reference and for the curated \QtRCC{} models. 

\subsection{Absence of tube-supported logical operators}
\label{sec:tubes}

In this subsection, we report a numerical search for logical operators supported in axis-aligned tube geometries. No logical operators are found in this search.

The no-string theorem of Sec.~\ref{sec:nostringmain} excludes fixed-width string segments on the infinite lattice, but not logical operators on periodic finite-size lattices whose width grows with system size.
Such objects are wide string-like supports: their transverse width can be as large as $L-1$, while the longitudinal extent is also $O(L)$.
One important case is a tube-supported logical operator, with support dimensions $w\times w\times L$.
In the curated models, our numerical search finds no axis-aligned tube-supported logical operator in the $x$, $y$, or $z$ direction.
Equivalently, within the scanned family, no logical operator is supported on an axis-aligned tube with width $w<L$. 

To search for such operators, we scan axis-aligned tubes of support $w\times w\times L$ in all three lattice directions.
Let $S$ denote the set of physical qutrits in the tube and $C$ its complement.
The dimension of the supported $X$-type quotient associated with $S$ is
\begin{equation}
d_X(S)
=
|S|
-\rank H_Z[:,S]
-\rank(H_X)
+\rank H_X[:,C].
\label{eq:tubeX}
\end{equation}
Here $|S|-\rank H_Z[:,S]$ is the dimension of $X$-type operators supported on $S$ that commute with all $Z$-checks.
Among $X$-type stabilizers, the subgroup that has a representative supported inside $S$ has dimension $\rank(H_X)-\rank H_X[:,C]$, obtained by projecting the $X$-check row span onto the complement $C$.
Equation~\eqref{eq:tubeX} subtracts this supported stabilizer subgroup from the supported charge-neutral subspace.
The corresponding $Z$-type quantity $d_Z(S)$ is obtained by exchanging $X$ and $Z$.
A tube-supported logical operator exists if and only if
\begin{equation}
d_X(S)>0
\qquad\text{or}\qquad
d_Z(S)>0.
\end{equation}

Within this axis-aligned tube family, the strongest test uses the widest tubes, with width
\begin{equation}
w=L-1.
\end{equation}
Every narrower axis-aligned tube lies inside a width-$(L-1)$ tube with the same orientation.
Therefore, if no logical operator is supported on any width-$(L-1)$ tube, then no logical operator is supported on any narrower axis-aligned tube.
It is therefore sufficient, for this axis-aligned tube diagnostic, to scan only the widest proper tubes.

The completed search covers all tube directions and positions in all curated models.
For every scanned support region and every curated model we find
\begin{equation}
d_X(S)=d_Z(S)=0.
\label{eq:tubezero}
\end{equation}
Thus no proper axis-aligned tube carries logical content in either CSS Pauli type.

This result  strengthens the membrane picture established in Sec.~\ref{sec:planes}.
Within the scanned axis-aligned tube family, the logical operators are not merely represented on planes; no support family strictly narrower than a membrane carries logical content.
In particular, the widest possible proper tube already carries no logical operator.

We finally note again that the tube scan still does not exclude every conceivable irregular support geometry.
For example, an irregular fractal support that is neither membrane-like nor tube-like is not ruled out by Eq.~\eqref{eq:tubezero}.

\section{Absence of self-similar fractal operators}
\label{sec:no.fractal}

Section~\ref{sec:numerics} reports membrane-supported logical representatives and no proper axis-aligned tube logicals in the curated \QtRCC{} collection~\cite{QTRCCRepository}.
A related important question is whether the self-similar fractal charge-growth mechanism of fracton codes such as Haah's Code 1 survives after the coefficients are made spatially nonuniform.
This diagnostic probes the same mechanism that underlies self-similar fractal logical operators in the uniform model.
Besides the ground-state degeneracy, it provides a direct test of how constrained spatial randomness affects the fractal structure and charge dynamics.

We test this by a concrete local charge-push protocol.
Starting from a single $XI$ operator applied to a ground state, we create four violated $Z$ checks.
We then keep one charge fixed at its original location and push the other three charges layer by layer in the $[111]$ direction using only $XI$ and $X^2I$ operators.
This procedure is deterministic once the initial step is fixed.
We record the number $N(\ell)$ of nonzero charges at layer $\ell$ and examine whether the pattern of $XI/X^2I$ operators forms a finite self-similar fractal, with only four charges at its corners at layer $\ell=3^n$.
We also repeat the same diagnostic for the $IX$ channel, using only $IX$ and $IX^2$ moves.

For reference, in the spatial uniform \QtHC{}, the charge pattern has a power-of-three fractal recurrence: at layers $\ell=3^n$, the process returns to four charges at the corners of the fractal support.
In contrast, the curated \QtRCC{} examples exhibit very distinct behavior.

\subsection{Layer-9 charge-count statistics}
\label{sec:chargepush-layer9}

\begin{figure}[th!]
\includegraphics[width=\columnwidth]{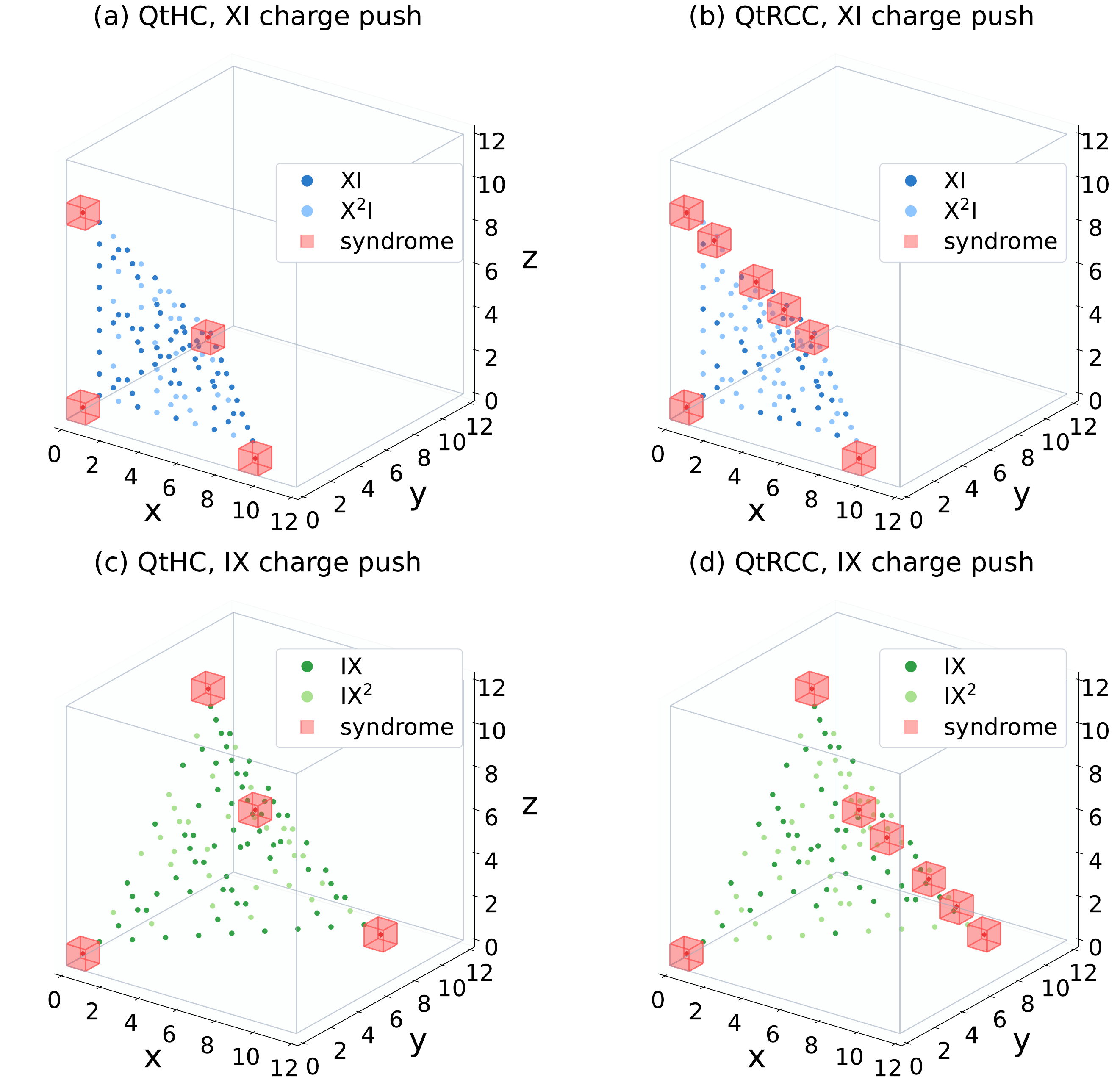}
\caption{Layer-9 charge-push result for the uniform \QtHC{} benchmark and a representative curated \QtRCC{} example \textsf{L12\_model\_01}~\cite{QTRCCRepository}.
(a,b) the charge push at layer 9 by the $XI/X^2I$ operators for  \QtHC{} and  \QtRCC{} example respectively. The $XI/X^2I$ operators are marked by blue dots, and red translucent cubes mark the violated $Z$-checks remaining at layer 9.
(c, d) the charge push at layer 9 by the $IX/X^2I$ operators, marked by green dots.
In both cases, the \QtHC{} returns to four charges in both channels at this power-of-three layer.
The \QtRCC{} model remains locally pushable but does not return to the same four-charge pattern, producing additional violated checks whose positions depend on the spatial coefficient field.}
\label{fig:chargepush-l9-types}
\end{figure}

\begin{figure}[th!]
\includegraphics[width=\columnwidth]{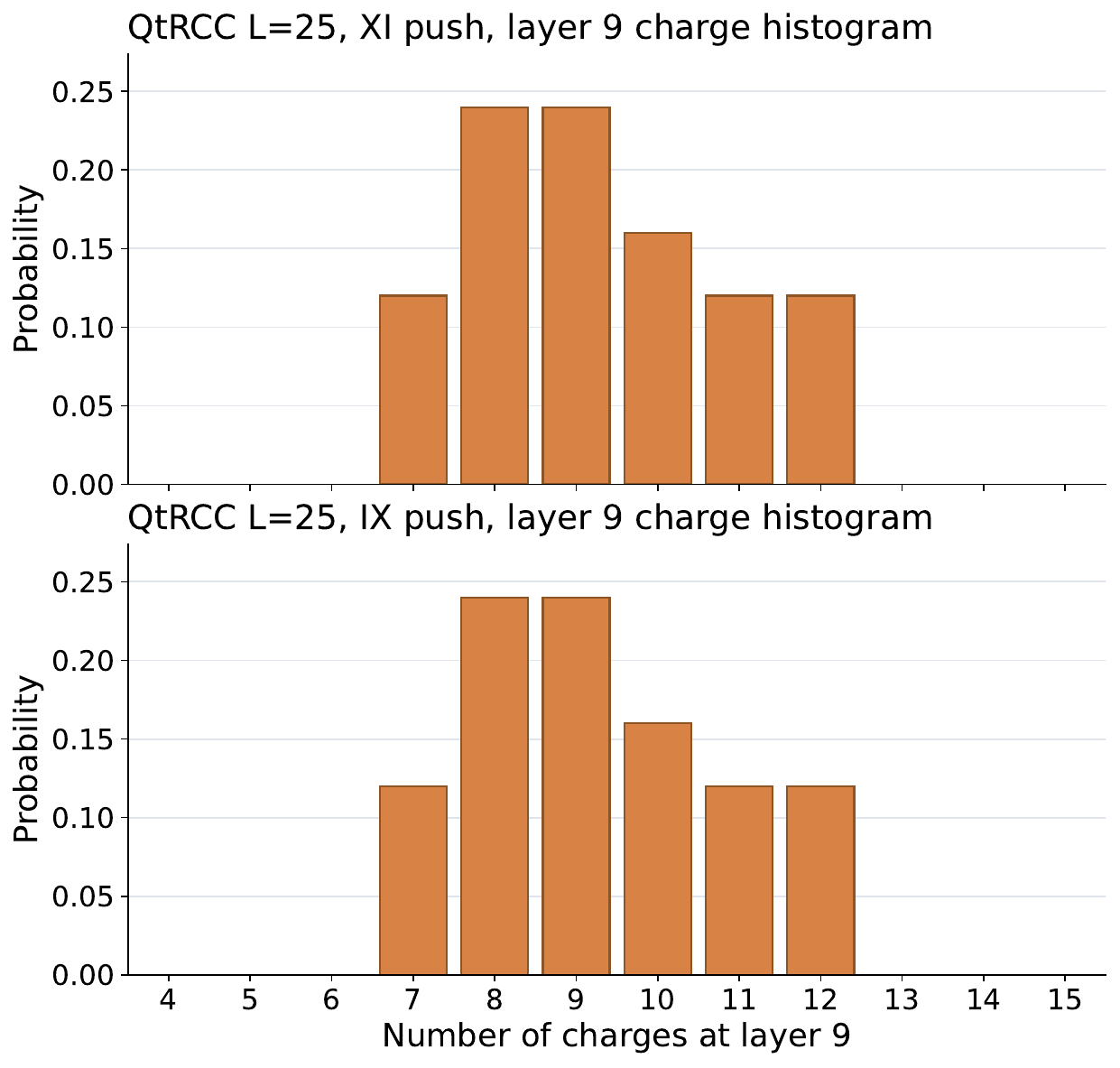}
\caption{Layer-9 charge-count statistics for the curated random \QtRCC{} example \textsf{L25\_model\_01}~\cite{QTRCCRepository} at $L=25$.
The upper panel shows the empirical probability distribution for the first-qutrit $XI$ push, and the lower panel shows the second-qutrit $IX$ push.
Each panel uses all $25^2$ distinct starting corners in the same model.
In the \QtHC{} benchmark, layer 9 is a power-of-three recurrence layer with four total charges.
In this \QtRCC{} instance, both channels are instead supported on 7 through 12 charges, with common mode 8 and mean approximately 9.2.}
\label{fig:layer9hist}
\end{figure}

Layer $9=3^2$ is the first nontrivial power-of-three recurrence scale beyond layer 3 in the \QtHC{} benchmark.
In \QtHC{}, the push returns to four charges at the fractal support corners at this layer.
We therefore first ask whether a representative curated \QtRCC{} example shows any comparable four-charge recurrence at layer 9.

The protocol is deterministic once the starting corner, coefficient field, and qutrit channel are fixed.
A single first-qutrit operator $X_{s,1}^t$ at $s=(x,y,z)$ violates four neighboring $Z$-checks, at
\begin{equation}
\begin{gathered}
s+(-1,-1,-1),\quad s+(0,-1,-1),\\
s+(-1,0,-1),\quad s+(-1,-1,0),
\end{gathered}
\end{equation}
with charges
\begin{equation}
(w_1t,\ u_Dt,\ u_Et,\ u_Ft)
\end{equation}
respectively, where the coefficients are evaluated at the corresponding anchors.
Since all coefficients are nonzero in $\F_3$, any non-reference charge $q$ at cube $r$ can be canceled uniquely by applying a first-qutrit operator at $r+(1,1,1)$ with exponent
\begin{equation}
t=-q/w_{1,r}.
\label{eq:pushrule}
\end{equation}
For the second-qutrit channel, $X_{s,2}^t$ instead creates charges
\begin{equation}
(w_2t,\ v_At,\ v_Bt,\ v_Ct),
\end{equation}
and the cancellation exponent is $t=-q/w_{2,r}$.
We keep one reference charge fixed and iterate the cancellation rule for every other charge, thereby pushing the syndrome front one layer forward along the $[111]$ direction.

Figure~\ref{fig:chargepush-l9-types} gives an example of the corresponding layer-9 charge configurations for both qutrit channels in \QtHC{} and in a \QtRCC{} model.
In the \QtHC{} benchmark, both $XI$ and $IX$ pushes leave the familiar four endpoint cube charges at the power-of-three layer.
In the \QtRCC{} example, the same deterministic cancellation rule produces additional violated $Z$ checks at layer 9, with different spatial patterns in the two qutrit channels.

Figure~\ref{fig:layer9hist} then quantifies the \QtRCC{} layer-9 charge-count distribution for the curated \QtRCC{} example \textsf{L25\_model\_01}~\cite{QTRCCRepository} at $L=25$.
The upper panel uses all $25^2$ distinct (due to the $[111]$ line symmetry) sampled starting corners in the first-qutrit $XI$ channel, and the lower panel uses all distinct starting corners in the second-qutrit $IX$ channel.

In both channels, after the layer-9 push, the charge counts lie in a compact 7-to-12-charge window, with no sampled start ending with charge count $4$, $5$, or $6$.
The two histograms have identical distributions 
\begin{equation}
	\Pr_{XI/IX}(N(9)=n) =
	\begin{cases}
		75/625 = 0.12, & n=7,\\
		150/625 = 0.24, & n=8,\\
		150/625 = 0.24, & n=9,\\
		100/625 = 0.16, & n=10,\\
		75/625 = 0.12, & n=11,\\
		75/625 = 0.12, & n=12,\\
		0, & \text{otherwise}.
	\end{cases}
\end{equation}
with mean
\begin{equation}
	\langle N_9\rangle_{XI}
	=
	\langle N_9\rangle_{IX}
	=
	9.28.
\end{equation}
Thus the layer-9 charge excess over the uniform four-charge recurrence is persistent across the sampled starting corners.

\subsection{Charge push and loss of power-of-three recurrence}
\label{sec:chargepush}

\begin{figure}[th!]
\includegraphics[width=\columnwidth]{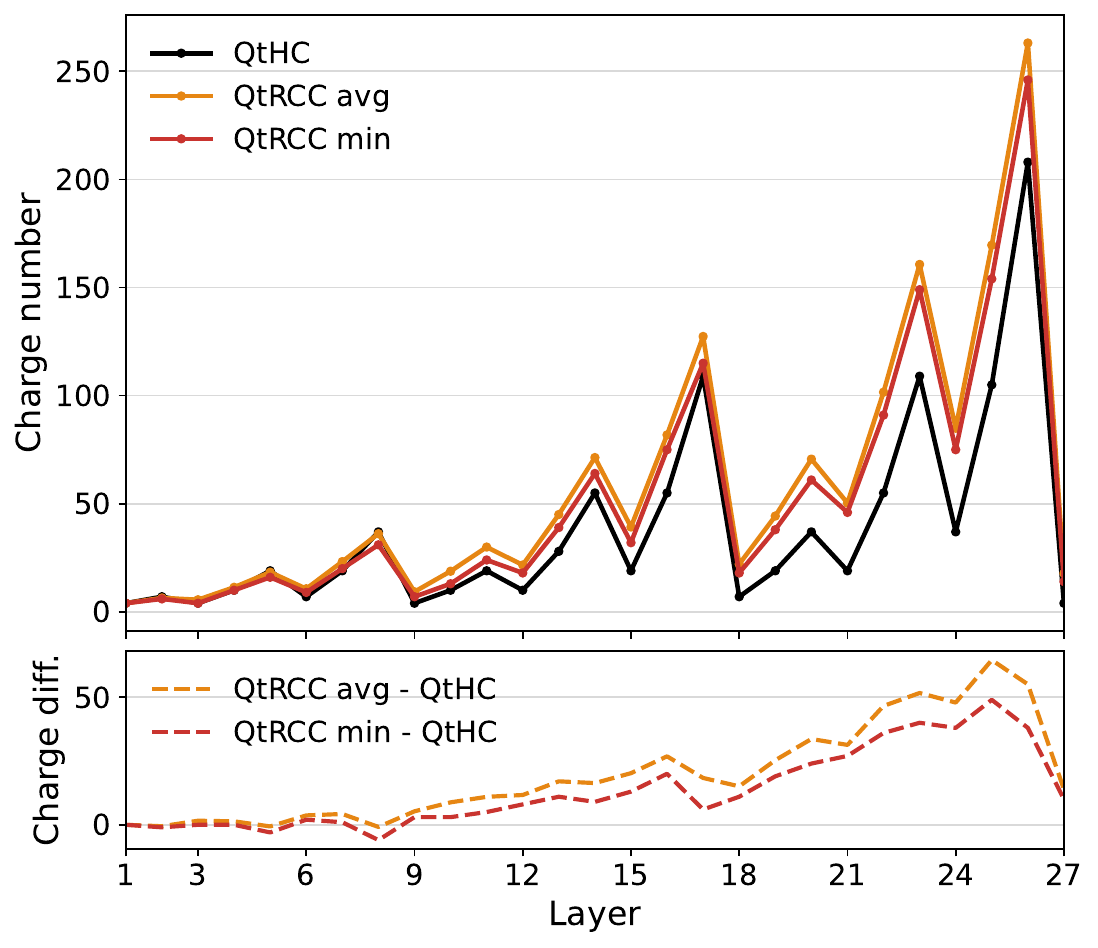}
\caption{Charge push in the uniform and constrained-random qutrit cubic codes at $L=25$.
Starting from a single first-qutrit $X$ operator, every non-reference $Z$-check charge is canceled by the unique next-layer $IX/IX^2$ move in qutrit arithmetic, pushing the moving syndrome front along $[111]$.
The upper panel shows the total number of nonzero charges for the uniform \QtHC{} benchmark and for one curated \QtRCC{} example, \textsf{L25\_model\_01}~\cite{QTRCCRepository}.
For \QtRCC{}, the average and layerwise minimum are taken over 100 sampled starting corners.
The lower panel subtracts the uniform result.
\QtHC{} returns to four charges at layers $1,3,9$, and $27$, while this \QtRCC{} instance does not: at layer $27$ the sample average is $17.78$ and the layerwise best sampled count is $14$.}
\label{fig:chargepush2}
\end{figure}

\begin{figure}[h!]
	\includegraphics[width=\columnwidth]{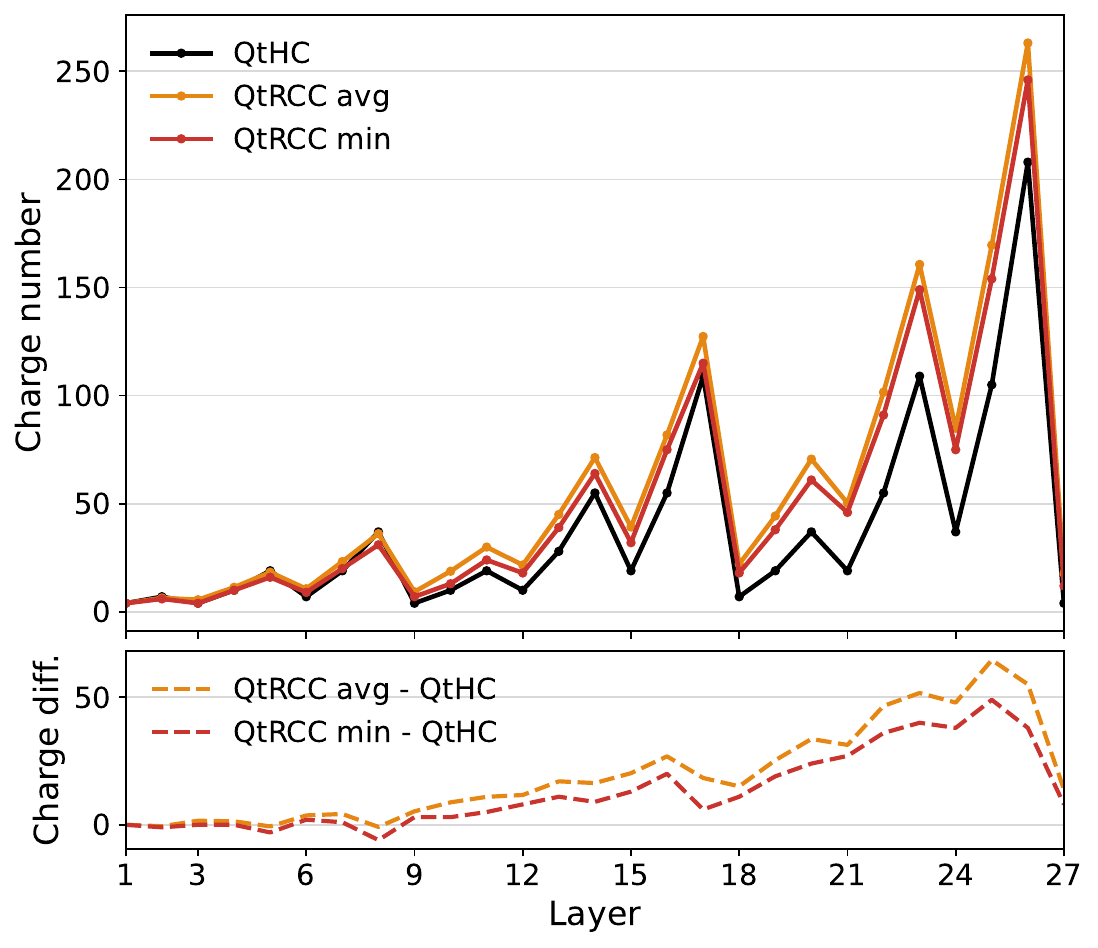}
	\caption{Charge push in the uniform and constrained-random qutrit cubic codes at $L=25$ using the $IX/IX^2$ operators.}
	\label{fig:chargepush}
\end{figure}

We next follow the same diagnostic through a larger range of push layers.
Let $N(\ell)$ be the total number of nonzero $Z$-check charges after the $\ell$th layer has been generated in the first-qutrit push.
The layer-$27=3^3$ data in Fig.~\ref{fig:chargepush} are computed on the periodic $L=25$ lattice and intentionally include the wraparound of the pushed syndrome front to see periodical boundary effect.
We use this as a finite-torus self-similar fractal diagnostic: the uniform \QtHC{} reference still returns to four charges under the same periodic arithmetic, whereas the constrained-random \QtRCC{} instance does not.

For the spatially uniform \QtHC{} benchmark, the exact power-of-three fractal recurrence occurs: in the $L=25$ calculation shown in Fig.~\ref{fig:chargepush},
\begin{equation}
N_{\QtHC}(1)=N_{\QtHC}(3)=N_{\QtHC}(9)=N_{\QtHC}(27)=4 .
\end{equation}
The total charge number grows between these special layers, but it repeatedly returns to the four charges when self-similar fractal support forms.

The \QtRCC{} model remains locally pushable deterministically, but the spatially varying coefficients change the number of charges at each layer.
Figure~\ref{fig:chargepush} compares \QtHC{} with one representative curated \QtRCC{} example, \textsf{L25\_model\_01}~\cite{QTRCCRepository}, using a all distinct $25^2$ starting points with $XI/X^2I$ charge push.
The orange curve is the average over starting corners, and the red curve is the layerwise minimum over all starting corners trajectories.
The two points on the minimum curve  may correspond to different starting points and cannot be realized in a single charge-pushing trajectory.

At the first layer the two models agree, because the local syndrome has the same one-plus-three structure.
The distinction appears at the recurrence scales.
At layer $9$, where \QtHC{} returns to four charges, the \QtRCC{} example has
\begin{equation}
	N_{\QtRCC}^{\rm avg}(9)=9.28,\qquad
	N_{\QtRCC}^{\rm min}(9)=7 .
\end{equation}
At layer $27$ the difference is larger:
\begin{equation}
	N_{\QtRCC}^{\rm avg}(27)=17.84,\qquad
	N_{\QtRCC}^{\rm min}(27)=14,
\end{equation}
whereas $N_{\QtHC}(27)=4$.
Thus even the best sampled starting corner at the power-of-three scale gains additional charges.
The lower panel of Fig.~\ref{fig:chargepush} subtracts the uniform result and shows that, after the first few layers, the \QtRCC{} charge count remains persistently above the uniform benchmark throughout the plotted window.
The largest average excess in layers $1\le\ell\le27$ is $64.60$ charges at layer $25$, and the largest layerwise-minimum excess is $49$ charges, also at layer $25$.

Figure~\ref{fig:chargepush} shows the result with $IXI/IX^2$ charge push. Interestingly, the corresponding $IX$-push data give the same layerwise average charge counts as the $XI$ channel over all $25^2$ distinct starts. The layerwise minima also agree through layer 26; the only difference in layers $1\le \ell\le 27$ occurs at layer 27, where the $IX$ minimum is 12 instead of the XI value $14$.

In this charge-push calculation, the tested self-similar fractal charge-creation process is absent in the representative \QtRCC{} model.
The local cancellation moves still exist, but the global power-of-three recurrence of \QtHC{} is not observed.
This means the self-similar fractal operators on a single sublattice cannot form a logical operator even when lattice size is compatible. 
This mirrors the ground-state-degeneracy study, which shows no arithmetic dependence of the \QtHC{} type.
Together, these diagnostics show the absence of the tested self-similar fractal recurrence in \QtRCC{}: neither the finite-size GSD data nor the charge-push process retains the regular power-of-three structure of the uniform reference model.
However, we note that they do not exclude every possible irregular non-self-similar fractal representative that mixes the two sublattice operators.

\section{Discussion and outlook}
\label{sec:discussion}

\subsection{Established properties}

We have introduced the qutrit random cubic code, a family of local qutrit CSS stabilizer Hamiltonians that keep the cube support structure of Haah's Code 1 while allowing the finite-field stabilizer exponents to vary in space.
This variation or randomness is highly constrained: on the torus the $Z$ checks obey a global stabilizer relation, and locally the $A$ and $B$ checks must satisfy the CSS commutation equations derived in Sec.~\ref{sec:model}. 
We have generated a curated set of models of $10 \le L \le 25$ with minimal GSD at each $L$. These models are publicly available at repository~\cite{QTRCCRepository}.

The main rigorous result is the no-string theorem for \QtRCC.
For every locally admissible \QtRCC{} coefficient field, there are no fixed-width logical string segments in Haah's sense.
The proof is generalized from the Haah's Code 1 proof and verifies that the required corner-erasing moves, good-edge-erasing moves, and exposed-edge confusing constraint remain valid over $\F_3$ because all relevant coefficients are nonzero and hence invertible.

The numerical calculations explore the properties of \QtRCC{} in greater depth.
For $10\le L\le25$, the smallest observed ground-state degeneracy exponent follows the simple parity pattern $k=2$ for odd $L$ and $k=4$ for even $L$.
Noncontractible plane-supported operators span the entire logical operator space.
The tube scan finds no proper axis-aligned tube logical operator, even for the widest possible proper tubes of width $L-1$.
Finally, the charge-push diagnostics show that the power-of-three self-similar fractal recurrence of the uniform \QtHC{} reference is absent in representative \QtRCC{} models.
Taken together, these results show that the constrained randomness in \QtRCC{} preserves the Haah-type no-string mechanism while suppressing the regular self-similar fractal mechanism of the translation-invariant reference model.

\subsection{Speculation}

The observations above leave one central question unresolved: \QtRCC{} may or may not admit irregular fractal logical operators that are neither fixed-width strings, nor axis-aligned tubes, nor the power-of-three self-similar fractals of \QtHC{}.
This question determines the likely large-scale memory behavior of the model.

If such irregular fractal representatives are absent, then the optimistic expectation is that the smallest logical operators are the plane-supported logicals observed numerically.
In that scenario the distance would scale as
\begin{equation}
d=\Theta(L^2),
\end{equation}
up to the parity-dependent prefactors seen in the plane-representative data.
Moreover, a local error process that builds a membrane-like logical operator would generically carry a one-dimensional boundary of violated checks.
This makes a linear energy barrier,
\begin{equation}
\Delta E=\Theta(L),
\end{equation}
a plausible outcome.
Such behavior would realize a promising combination of properties for  a family of  QEC: Haah-type absence of strings, membrane logical operators, and a potentially linear barrier on a regular three-dimensional cubic geometry.

If irregular fractal logical representatives do exist, the conclusion would be different but still physically important.
Their existence would show that constrained spatial variation can destroy the algebraic self-similar recurrence of the uniform model while generating a more disordered form of fractal logical geometry.
Such operators would not be captured by translation-invariant polynomial methods. 
They would instead point to a new regime of spatially nonuniform stabilizer Hamiltonians, where fractal logical structure is controlled by local finite-field constraints rather than by exact translation symmetry.

\subsection{Open problems and future directions}

There are many open problems revolving around the \QtRCC{} family.
The first one is to determine which of these two scenarios above is true.
This requires either proving a lower bound that rules out non-plane logical representatives below membrane scale, or explicitly constructing irregular logical representatives of smaller support.
Both directions require methods beyond the plane and tube scans reported here, such as global Pauli-support optimization or analytic lower-bound techniques adapted to spatially varying stabilizer coefficients.

A second, related problem is the energy barrier and memory lifetime~\cite{Brown2016,BravyiHaah2013}.
The no-string theorem rules out constant-width string processes, and the observed plane representatives suggest membrane-like logical operators, but the optimal charge-creation history is not known.
Establishing whether the barrier is logarithmic, linear, or governed by an irregular fractal process is essential for assessing the passive-memory potential of \QtRCC.
Closely related questions include finite-temperature dynamics, free-energy barriers, and the construction of efficient decoders.

Furthermore, it is important to broaden the construction beyond the line-symmetric ensemble used in the numerical study.
Our sampled models are a $[111]$ line-symmetric subcollection of the full locally admissible family.
Constructing and sampling more general three-dimensional admissible coefficient fields would test whether the parity pattern, plane representatives, tube exclusion, and absence of self-similar recurrence persist without this restriction.
It would also clarify which features extend to related prime-qudit cubic models, Haah-code generalizations, and manifold extensions~\cite{LeeHuChoWatanabe2025,TianWang2019,TianSampertonWang2020}.

More broadly, \QtRCC{} suggests a route toward constrained-random local stabilizer codes on regular geometries, which potentially have improved (partial) self-correcting properties. 
It seems such QECs have been rarely studied so far.
The regular geometry and local stabilizer checks can make rigorous  proofs, direct numerical diagnostics, and potential experimental realization easier. 
On the many-body side, the curated examples combine finite, weakly size-dependent torus degeneracy of topological order with local fracton-like charge immobility inherited from the fracton codes.
This combination suggests a class of spatially nonuniform stabilizer phases whose classification is not captured by translation-invariant topological order or fracton orders alone.
Developing algebraic and field-theoretic frameworks, constructing  classification schemes for such phases, and exploring their quantum dynamics, are important directions for future work.

\begin{acknowledgments}
The author thanks Claudio Castelnovo, Dominic J. Williamson, Takumi Fukushima, Ryohei Kobayashi, and Masaki Oshikawa for helpful discussions.
This work was supported by JSPS KAKENHI Grant No.~26K17090, Grant-in-Aid for Early-Career Scientists.
\end{acknowledgments}

\section*{Data availability}
The data set for reproducing calculations on the curated models is available in the QTRCC repository at \url{https://github.com/hanyanphysics/QTRCC}~\cite{QTRCCRepository}.
It contains the complete definitions of the 374 curated periodic-lattice \QtRCC{} models with system sizes $10\le L\le25$.
Each model directory contains \texttt{z\_checks.csv} and \texttt{x\_checks.csv}, which specify the local $Z$-check and $X$-check parameters needed to reconstruct the CSS checks.
The accompanying \texttt{manifest.csv} lists the models and records their basic code parameters, including $\rank(H_Z)$, $\rank(H_X)$, and $k$.
The repository \texttt{README} gives the coordinate, qutrit-labeling, and matrix-reconstruction conventions.
The repository also includes scripts to check sampled models against the local topology and commutation constraints.

\bibliographystyle{apsrev4-2}
\bibliography{qtrcc}
\onecolumngrid
\appendix

\section{Local commutation equations and reduced construction}
\label{app:commutation}

This appendix records the algebraic commutation constraints used in the periodic finite-size construction.
Let $\delta=r_B-r_A$ be the displacement from a $Z$-check anchor to an $X$-check anchor.
14 of the 27 local overlaps give nontrivial commutation relations.
With the $B$ variables evaluated at the central $X$ anchor and the $A$ variables evaluated at the displaced $Z$ anchor, the equations are:
\begin{align}
\delta=(-1,0,1):&\quad v_C r_A+u_Ds_F=0,\label{eq:raw1}\\
\delta=(-1,1,0):&\quad v_B r_A+u_Ds_E=0,\\
\delta=(0,-1,1):&\quad v_C r_B+u_Es_F=0,\\
\delta=(0,0,0):&\quad v_A r_A+v_B r_B+v_C r_C+u_Ds_D+u_Es_E+u_Fs_F=0,\\
\delta=(0,0,1):&\quad v_Ct_2+2w_1s_F=0,\\
\delta=(0,1,-1):&\quad v_Br_C+u_Fs_E=0,\\
\delta=(0,1,0):&\quad v_Bt_2+2w_1s_E=0,\\
\delta=(0,1,1):&\quad u_Dt_1+w_2r_A=0,\\
\delta=(1,-1,0):&\quad v_Ar_B+u_Es_D=0,\\
\delta=(1,0,-1):&\quad v_Ar_C+u_Fs_D=0,\\
\delta=(1,0,0):&\quad v_At_2+2w_1s_D=0,\\
\delta=(1,0,1):&\quad u_Et_1+w_2r_B=0,\\
\delta=(1,1,0):&\quad u_Ft_1+w_2r_C=0,\\
\delta=(1,1,1):&\quad 2w_1t_1+w_2t_2=0.\label{eq:raw14}
\end{align}
Multiplying any equation by the nonzero scalar $2\in\F_3$ gives an equivalent constraint.

The seven anchor equations solve the $X$ field in terms of the $A$ field.
With the shorthand in Eq.~\eqref{eq:alpha_beta_rho}, the solution is
\begin{equation}
r_A=-\alpha_Dt_1,\quad r_B=-\alpha_Et_1,\quad r_C=-\alpha_Ft_1,
\end{equation}
\begin{equation}
\begin{aligned}
t_2&=\rho t_1, & s_D&=\rho\beta_A t_1, & s_E&=\rho\beta_B t_1, & s_F&=\rho\beta_C t_1.
\end{aligned}
\end{equation}
Substituting into the remaining seven equations gives the reduced $A$-field constraints
\begin{align}
C_1&=u_{D,(-1,0,1)}\rho\beta_C-v_{C,(-1,0,1)}\alpha_D=0,\\
C_2&=u_{D,(-1,1,0)}\rho\beta_B-v_{B,(-1,1,0)}\alpha_D=0,\\
C_3&=u_{E,(0,-1,1)}\rho\beta_C-v_{C,(0,-1,1)}\alpha_E=0,\\
C_4&=u_{F,(0,1,-1)}\rho\beta_B-v_{B,(0,1,-1)}\alpha_F=0,\\
C_5&=u_{E,(1,-1,0)}\rho\beta_A-v_{A,(1,-1,0)}\alpha_E=0,\\
C_6&=u_{F,(1,0,-1)}\rho\beta_A-v_{A,(1,0,-1)}\alpha_F=0,\\
C_7&=u_{D,(0,0,0)}\rho\beta_A+u_{E,(0,0,0)}\rho\beta_B
+u_{F,(0,0,0)}\rho\beta_C\\
&\quad -v_{A,(0,0,0)}\alpha_D-v_{B,(0,0,0)}\alpha_E
-v_{C,(0,0,0)}\alpha_F=0.
\end{align}
Together with topology and the nonzero-exponent condition, these equations define the periodic finite-size coefficient constraints.

\section{Detailed proof of the no-string theorem}
\label{app:nostringproof}

Theorem~\ref{thm:nostring} is stated in Sec.~\ref{sec:nostringmain}.

\setcounter{theorem}{0}
\begin{theorem}[No fixed-width logical strings] 
	Consider a locally admissible \QtRCC{} coefficient field on $\Z^3$.
	Then there is a finite function $\phi(w)$ such that every nontrivial logical string segment of width $w$ has length at most $\phi(w)$.
	With the effective-width convention $W=w+4$, the proof gives the conservative bound
	\begin{equation}
		\phi(w)\le 27W+48=27(w+4)+48 . 
	\end{equation}
	The same statement holds on the periodic lattice $\Lambda_L$ for string segments whose support and anchor collar are contained in a contractible no-wraparound region.
	It does not forbid noncontractible plane-supported logical operators on the periodic lattice.
\end{theorem}

This appendix gives the detailed proof.
The proof follows the support-geometry deformation and erasing argument for Haah's Code 1~\cite{Haah2011,HaahThesis2013}; the model-dependent part is the qutrit coefficient algebra.
We first recall Haah's terminology and proof structure, then verify that the three local ingredients used in the main-text sketch remain valid for \QtRCC.
The argument is local and deterministic: it uses the cube support pattern of Haah's Code 1, the reflected CSS structure of Eqs.~\eqref{eq:Zcheck} and \eqref{eq:Xcheck}, local $A$--$B$ commutation, and the assumption that all coefficients are nonzero.
It does not use the topology identity, the $[111]$ line symmetry, or any property of the sampling procedure.

\subsection{Haah logical string segments}
\label{subsec:haah-string-definition}

We use Haah's definition of a logical string segment.
A width-$w$ anchor is a cube of $w^3$ lattice sites.
Let $\Omega_1$ and $\Omega_2$ be two congruent anchor cubes and let $P$ be a finite Pauli operator.
The triple
\begin{equation}
\zeta=(P,\Omega_1,\Omega_2)
\end{equation}
is a string segment if every stabilizer check that acts trivially on both anchors commutes with $P$.
Stabilizer violations are therefore allowed only at the two anchors.
If the displacement from $\Omega_1$ to $\Omega_2$ is $(a,b,c)$, then the length of the segment is
\begin{equation}
\ell(\zeta)=|a|+|b|+|c|,
\end{equation}
and its width is $w$.

Two string segments with the same anchors are equivalent if their Pauli operators differ by multiplication by finitely many stabilizer checks.
A representative is connected if the site support of its Pauli operator, together with the two anchors, contains a lattice path joining the anchors.
Here a site is counted in the support if either of its two qutrits carries a nonzero Pauli exponent.
The segment is nontrivial if every equivalent representative remains connected.
Equivalently, a segment is disconnected if multiplication by a finite product of stabilizer checks gives a representative whose support has no path connecting the two anchors.
Let $\phi(w)$ be the maximum length of a nontrivial width-$w$ string segment, allowing the value $\infty$.
A stabilizer model is free of string logical operators if
\begin{equation}
\phi(w)<\infty
\end{equation}
for every finite $w$.

Because every cube check has diameter one, we use the enlarged anchor collar
\begin{equation}
\Omega^+=\{s:\operatorname{dist}_\infty(s,\Omega_1\cup\Omega_2)\le 1\}.
\label{eq:anchor-collar}
\end{equation}
All check multiplications used below are supported away from $\Omega^+$.
Thus the deformation changes only the representative of the string segment and never changes the allowed anchor regions.

\subsection{Haah's Code 1 proof structure}
\label{subsec:haah-logic-flow}

Haah's Code 1 no-string proof is by contradiction: one assumes a nontrivial width-$w$ logical string segment and then constructs an equivalent disconnected representative.
The geometric part of the proof is support-contained, meaning that the checks used in the deformation do not create support outside the chosen string region or inside the anchor collar.

For Haah's Code 1, Haah established a support-contained geometric reduction using corner erasing, good-edge erasing, and the exposed-edge confusing constraint.
We import this geometric reduction.
The model-dependent task here is not to rederive the support geometry, but to verify that each local algebraic move used in that reduction remains valid after the binary coefficients are replaced by nonzero spatially varying elements of $\F_3$.
These verifications are given in the next three subsections.

For example, a flat corridor directed along $\hat y$ can be reduced, away from the anchors, to two rectangles of thickness one,
\begin{align}
R_z&=\{(x,y,z):0\le x\le W-1,\ y_0\le y\le y_1,\ z=0\},
\nonumber\\
R_x&=\{(x,y,z):x=0,\ y_0\le y\le y_1,\ 0\le z\le W-1\},
\label{eq:two-rectangle-normal-form}
\end{align}
where $W=w+4$ is a conservative effective width.

Once the support has been reduced to thickness-one rectangles, the final step in Haah's proof is the exposed-edge confusing constraint.
We describe the argument for an $X$-type Pauli component; the $Z$-type argument follows by exchanging $X$ and $Z$ and using the reflected check pattern.
After the support is already in the two-rectangle normal form, one studies an exposed straight boundary edge of a thin rectangle.
In Haah's Code 1, the local commutation-with-$Z$-checks constraints along that edge mean that every pair of neighboring sites on the edge has only four binary possibilities
\begin{equation}
\mathrm{II}\!-\!\mathrm{II},
\quad \mathrm{XI}\!-\!\mathrm{II},
\quad \mathrm{IX}\!-\!\mathrm{XI},
\quad \mathrm{XX}\!-\!\mathrm{XI}.
\label{eq:haah-binary-confusing-list}
\end{equation}
Consistency along the edge for consecutive three sites forces the transitions
\begin{equation}
\mathrm{II}\to \mathrm{II},
\quad
\mathrm{XI}\to \mathrm{II},
\quad
\mathrm{IX}\to \mathrm{XI}\to \mathrm{II},
\quad
\mathrm{XX}\to \mathrm{XI}\to \mathrm{II}.
\label{eq:haah-binary-confusing-transition}
\end{equation}
Thus the exposed edge is eventually identity away from the anchor.
Haah calls such an exposed-edge condition a confusing constraint.
Repeating the same argument row by row gives a triangular erased region of the support: each time the induction moves one row inward, the forced identity interval begins one site farther from the anchor than it did on the previous edge.
If the flat segment is long compared with its width, the erased region can cut through the segment, turning it into two disconnected pieces.
This concludes the proof by contradiction.

This geometric procedure is summarized in Fig.~\ref{fig:nostringproof} in Sec.~\ref{sec:nostringmain}.
We now prove in detail that the same corner-erasing, good-edge-erasing, and confusing-constraint mechanisms remain valid for the \QtRCC{} checks.

\subsection{Corner erasing for qutrit coefficients}
\label{subsec:qtrcc-corner-erasing}

\begin{lemma}[Corner erasing for \QtRCC]
\label{lem:qtrcc-corner-erasing}
Every support-contained corner-erasing move used in Haah's Code 1 proof has a valid \QtRCC{} analogue.
The only local cube corner that cannot serve as an erasing corner is the identity operator corner.
\end{lemma}

\begin{proof}
We prove the statement for an $X$-type Pauli operator; the $Z$-type statement is obtained by exchanging $X$ and $Z$ and using the reflected check pattern.
Let $p$ be an exposed target site away from $\Omega^+$, and write the local $X$ exponent at $p$ as
\begin{equation}
\xi(p)=(a,b)\in\F_3^2 .
\end{equation}
Suppose a $Z$-check touches the current possible support only at $p$, with local $Z$-corner vector
\begin{equation}
z=(z_1,z_2)\in\F_3^2 .
\end{equation}
Commutation with this check gives
\begin{equation}
z\cdot \xi(p)=z_1a+z_2b=0 .
\label{eq:corner-local-constraint}
\end{equation}
If $z=(0,0)$, this equation gives no restriction.
This is the identity corner, and it cannot be used as an erasing corner.

Consider first a one-qutrit corner, for example
\begin{equation}
z=(0,\alpha),\qquad \alpha\in\F_3^* .
\end{equation}
Equation~\eqref{eq:corner-local-constraint} gives $b=0$, so the remaining allowed local exponent is $(a,0)$.
In the corresponding Haah corner-erasing move, the $X$-check has the same target site and a one-qutrit target corner in the remaining direction,
\begin{equation}
x=(\beta,0),\qquad \beta\in\F_3^* .
\end{equation}
This is a support-pattern statement: the $X$-check is the reflected check used in the same Haah move.
Multiplying by the check power $\lambda=-a/\beta$ cancels the target exponent.
The case $z=(\alpha,0)$ is the same with the two qutrit components exchanged.

Now consider a mixed two-qutrit corner,
\begin{equation}
z=(w_1,w_2),\qquad w_1,w_2\in\F_3^* .
\end{equation}
The allowed local exponent lies in the one-dimensional kernel
\begin{equation}
K_z=\{\xi\in\F_3^2:z\cdot \xi=0\}.
\end{equation}
For the mixed $G$ corner of a $Z$-check, use the $X$-check whose mixed $O$ corner coincides with the target site.
Its local $X$ vector is
\begin{equation}
x=(t_1,-t_2),\qquad t_1,t_2\in\F_3^* .
\end{equation}
The local mixed-corner $A$--$B$ commutation equation is
\begin{equation}
w_1t_1-w_2t_2=0,
\label{eq:mixed-corner-commutation}
\end{equation}
so $z\cdot x=0$.
Hence $x\in K_z$.
Since $x\ne0$ and $K_z$ is one-dimensional, $x$ spans $K_z$.
Therefore any allowed target exponent is $\xi(p)=\mu x$, and multiplying by the check power $-\mu$ erases it.

Finally, in a zero-kernel corner move two independent $Z$-checks isolate the same target site.
Their local corner vectors $z,z'\in\F_3^2$ are linearly independent, so the two equations
\begin{equation}
z\cdot \xi(p)=0,\qquad z'\cdot \xi(p)=0,
\end{equation}
force $\xi(p)=0$.
No check multiplication is then needed.

All checks used here have the same target site and the same non-target corner positions as in the corresponding move in Haah's Code 1.
Therefore the support-containment and anchor-avoidance parts of the move are unchanged.
\end{proof}

\subsection{Good-edge erasing for qutrit coefficients}
\label{subsec:qtrcc-good-edge-erasing}

The good $Z$ edges used to erase $X$-type support in the imported Haah deformation are the edges joining a qutrit-2-only $Z$ corner to a qutrit-1-only $Z$ corner:
\begin{equation}
\{C\!-\!D,\ C\!-\!E,\ B\!-\!D,\ B\!-\!F,\ A\!-\!E,\ A\!-\!F\}.
\label{eq:qtrcc-good-edges}
\end{equation}
At the two endpoints of such an edge, the $Z$-corner vectors have the form
\begin{equation}
(0,\alpha),
\qquad
(\beta,0),
\qquad
\alpha,\beta\in\F_3^* .
\label{eq:good-edge-endpoint-vectors}
\end{equation}
This is the qutrit version of the independence condition for a good edge.

\begin{lemma}[Good-edge kernel for qutrit coefficients]
\label{lem:qtrcc-good-edge-kernel}
Let an $X$-type Pauli have endpoint exponents
\begin{equation}
\xi(p)=(a,b),
\qquad
\xi(q)=(c,d)
\end{equation}
on a good $Z$ edge with endpoint $Z$ vectors $(0,\alpha)$ and $(\beta,0)$.
Then the local commutation constraint is
\begin{equation}
\alpha b+\beta c=0,
\label{eq:qtrcc-good-edge-constraint}
\end{equation}
and its kernel is spanned by
\begin{equation}
((1,0),(0,0)),\qquad
((0,0),(0,1)),\qquad
((0,\beta),(-\alpha,0)).
\label{eq:qtrcc-good-edge-kernel-basis}
\end{equation}
\end{lemma}

\begin{proof}
The commutation phase is
\begin{equation}
(0,\alpha)\cdot(a,b)+(\beta,0)\cdot(c,d)=\alpha b+\beta c .
\end{equation}
Since $\alpha$ and $\beta$ are nonzero, Eq.~\eqref{eq:qtrcc-good-edge-constraint} is one nontrivial linear equation in four variables.
The three vectors in Eq.~\eqref{eq:qtrcc-good-edge-kernel-basis} satisfy the equation and are linearly independent, so they form a basis of the three-dimensional kernel.
\end{proof}

The first two basis vectors in Eq.~\eqref{eq:qtrcc-good-edge-kernel-basis} can be removed by the corner-erasing moves of Lemma~\ref{lem:qtrcc-corner-erasing}.
The third basis vector is the correlated direction along the good edge.
This direction is supplied by local $A$--$B$ commutation.

Indeed, suppressing anchor shifts and writing only the coefficients at the two overlapping checks, the six two-term commutation equations are
\begin{equation}
\begin{aligned}
C-D:&\quad v_C r_A+u_D s_F=0,\\
B-D:&\quad v_B r_A+u_D s_E=0,\\
C-E:&\quad v_C r_B+u_E s_F=0,\\
B-F:&\quad v_B r_C+u_F s_E=0,\\
A-E:&\quad v_A r_B+u_E s_D=0,\\
A-F:&\quad v_A r_C+u_F s_D=0 .
\end{aligned}
\label{eq:qtrcc-good-edge-two-term-commutation}
\end{equation}
For example, on the $C-D$ edge the $Z$ vectors are $(0,v_C)$ and $(u_D,0)$, while the overlapping $X$-check restricts to $(0,r_A)$ and $(s_F,0)$.
The equation $v_C r_A+u_D s_F=0$ says exactly that this $X$-check restriction lies in the correlated one-dimensional kernel spanned by $((0,u_D),(-v_C,0))$.
Since every coefficient is nonzero, the restriction is a nonzero multiple of that vector.
The other five edges are identical.

Thus the allowed two-site $X$ exponent on a good $Z$ edge has exactly three local freedoms.
Two are endpoint directions removed by corner erasing.
The third is the correlated edge direction, and the corresponding two-term $A$--$B$ commutation equation shows that it is the restriction of a two-site overlapping $X$ check.
Therefore the same local checks used in Haah's good-edge move span the entire allowed qutrit edge kernel.

\begin{lemma}[Good-edge erasing]
\label{lem:qtrcc-good-edge-erasing}
Every support-contained good-edge erasing move used in Haah's Code 1 proof has a valid \QtRCC{} analogue.
\end{lemma}

\begin{proof}
By Lemma~\ref{lem:qtrcc-good-edge-kernel}, the two-site restriction allowed by commutation with the relevant $Z$-check is the span of two endpoint operators and one correlated edge operator.
The endpoint operators are erased by Lemma~\ref{lem:qtrcc-corner-erasing}.
The correlated edge operator is the support of an edge-overlapping $X$-check by Eq.~\eqref{eq:qtrcc-good-edge-two-term-commutation}.
Hence the support-contained $X$-checks span the entire allowed edge kernel.
Multiplying by suitable check powers erases the edge.
Because the check supports occupy the same cube corners as in Haah's Code 1, the support-containment condition is unchanged.
\end{proof}

The reflected statement for erasing $Z$-type support follows by exchanging $X$ and $Z$ and using the reflected check pattern.

Lemmas~\ref{lem:qtrcc-corner-erasing} and \ref{lem:qtrcc-good-edge-erasing} are the coefficient-weighted lift of Haah's local erasing moves.
Therefore the support deformation used in Haah's Code 1, first to three flat coordinate segments and then to the two-rectangle normal form of Eq.~\eqref{eq:two-rectangle-normal-form}, can be performed in \QtRCC{} with the same support geometry.

\subsection{Confusing constraints on a thin rectangle}
\label{subsec:qtrcc-confusing-constraint}

It remains to generalize the last stage of Haah's proof: the confusing constraint on the exposed edge of a thickness-1 flat rectangle segment of the operator.
We first discuss an $X$-type Pauli operator.
The corresponding statement for a $Z$-type Pauli operator is obtained by exchanging $X$ and $Z$ and using the reflected check pattern.

Consider a $\hat y$-directed flat segment in the two-rectangle normal form \eqref{eq:two-rectangle-normal-form}.
Let
\begin{equation}
p_j=(W-1,j,0),\qquad \xi(p_j)=(a_j,b_j).
\end{equation}
These sites form the exposed edge of $R_z$ with largest $x$ and smallest $z$.
For $j$ outside the enlarged anchors, take
\begin{equation}
r_j=p_j-C,
\qquad
r'_j=p_{j+1}-B .
\end{equation}
The check $A_{r_j}$ has its $C$ corner at $p_j$ and its $D$ corner at $p_{j+1}$, so the two relevant $Z$ vectors are
\begin{equation}
(0,v_C(r_j)),\qquad (u_D(r_j),0).
\end{equation}
All other nonidentity corners of $A_{r_j}$ lie outside $R_z\cup R_x$ in the flat normal form.
Hence commutation with the $X$-type operator gives
\begin{equation}
v_C(r_j)b_j+u_D(r_j)a_{j+1}=0 .
\label{eq:qtrcc-Rz-confusing-first}
\end{equation}
The check $A_{r'_j}$ has its $B$ corner at $p_{j+1}$ and its $O$ corner at $p_j$.
The $O$ corner of a $Z$-check is identity, and all other nonidentity corners lie outside $R_z\cup R_x$.
Therefore
\begin{equation}
v_B(r'_j)b_{j+1}=0 .
\label{eq:qtrcc-Rz-confusing-second}
\end{equation}

Similarly, on the exposed edge of $R_x$ let
\begin{equation}
q_j=(0,j,W-1),\qquad \xi(q_j)=(\tilde a_j,\tilde b_j).
\end{equation}
Take
\begin{equation}
s_j=q_j-A,
\qquad
s'_j=q_{j+1}-B .
\end{equation}
The check $A_{s_j}$ has its $A$ corner at $q_j$ and its $F$ corner at $q_{j+1}$, giving
\begin{equation}
v_A(s_j)\tilde b_j+u_F(s_j)\tilde a_{j+1}=0 .
\label{eq:qtrcc-Rx-confusing-first}
\end{equation}
The check $A_{s'_j}$ has its $B$ corner at $q_{j+1}$, while its $O$ corner at $q_j$ is identity and the remaining nonidentity corners are outside $R_z\cup R_x$.
Hence
\begin{equation}
v_B(s'_j)\tilde b_{j+1}=0 .
\label{eq:qtrcc-Rx-confusing-second}
\end{equation}
All coefficients in these equations are nonzero by admissibility.
For the exposed edge of $R_z$, the two local constraints can be summarized as
\begin{center}
\begin{tabular}{c c c c}
\toprule
check & edge sites touched & nonzero corner vectors & constraint \\
\midrule
$A_{p_j-C}$ & $p_j,p_{j+1}$ & $(0,v_C),(u_D,0)$ & $v_C b_j+u_D a_{j+1}=0$ \\
$A_{p_{j+1}-B}$ & $p_{j+1}$ & $(0,v_B)$ & $v_B b_{j+1}=0$ \\
\bottomrule
\end{tabular}
\end{center}
where the coefficients are evaluated at the check anchors displayed in the first column.
The table makes explicit the triangular structure of the confusing constraint: one check touches two consecutive edge sites, while the shifted check touches only the second site because the $O$ corner of a $Z$-check is identity.

\begin{lemma}[Confusing constraint for qutrit coefficients]
\label{lem:qtrcc-confusing}
Suppose an exposed-edge sequence satisfies
\begin{equation}
\alpha_j b_j+\beta_j a_{j+1}=0,
\qquad
\gamma_j b_{j+1}=0,
\end{equation}
with $\alpha_j,\beta_j,\gamma_j\in\F_3^*$.
Then the sequence is eventually zero: whenever the shifted equations are available,
\begin{equation}
b_{j+1}=0,
\qquad
a_{j+2}=0,
\qquad
b_{j+2}=0 .
\end{equation}
\end{lemma}

\begin{proof}
Since $\gamma_j\ne0$, the second equation gives
\begin{equation}
b_{j+1}=0 .
\end{equation}
Shifting the first equation by one site gives
\begin{equation}
\alpha_{j+1}b_{j+1}+\beta_{j+1}a_{j+2}=0 .
\end{equation}
Because $b_{j+1}=0$ and $\beta_{j+1}\ne0$, this implies
\begin{equation}
a_{j+2}=0 .
\end{equation}
Shifting the second equation gives $b_{j+2}=0$.
Therefore the site $p_{j+2}$ carries zero $X$ exponent.
Repeating the same argument shows that the exposed-edge sequence is identically zero beyond a bounded distance from the anchor.
This is the qutrit version of Haah's binary transition rule \eqref{eq:haah-binary-confusing-transition}.
\end{proof}

We have established the qutrit analogue of Haah's binary confusing constraint.
The exposed edge is not cleaned by freely multiplying checks along the edge.
Rather, the local commutation equations themselves are triangular: one equation forces a second-qutrit exponent on the next site to vanish, and the shifted two-site equation then forces the first-qutrit exponent one site farther along the edge to vanish.
Because all coefficients are nonzero in $\F_3$, the same eventual-vanishing operator conclusion holds with spatially varying coefficients.

Repeating the argument row by row gives the triangular erased region.
When one moves one row inward from an exposed edge, the additional non-target terms in the selected check equations lie only in the previous row, with a bounded shift along the $y$ direction.
Once the previous row has already vanished outside a distance $D$ from the anchors, the next row again obeys the exposed-edge equations outside distance $D+O(1)$, and Lemma~\ref{lem:qtrcc-confusing} forces it to vanish after another bounded transient.
The following row-by-row estimate is the same geometric estimate used in Haah's proof, applied after replacing the binary confusing constraint by Lemma~\ref{lem:qtrcc-confusing}.
The constants are deliberately conservative; no optimized constant is needed for the no-string theorem.
With this conservative estimate, row $m$ is zero outside distance $4m+4$ from either anchor.
Across at most $W$ rows on the two rectangles, all flat-segment support is confined within distance
\begin{equation}
D_{\rm flat}=4W+8
\label{eq:qtrcc-Dflat}
\end{equation}
from either anchor along the direction of the segment.
Therefore a flat segment of length
\begin{equation}
L_{\rm flat}=W+2D_{\rm flat}=9W+16
\label{eq:qtrcc-Lflat}
\end{equation}
has an empty middle interval and is equivalent to a disconnected representative.
The constants in Eqs.~\eqref{eq:qtrcc-Dflat} and \eqref{eq:qtrcc-Lflat} are not optimized; only finiteness at fixed width is needed.

\subsection{Proof of the no-string theorem}
\label{subsec:qtrcc-nostring-proof}

Having shown that Haah's local moves and the confusing constraint all have valid \QtRCC{} analogues, we can now complete the proof of Theorem~\ref{thm:nostring}.

\begin{proof}[Proof of Theorem~\ref{thm:nostring}]
First consider an $X$-type logical string segment.
By Haah's Code 1 support-deformation argument, the segment can be reduced, using support-contained corner and good-edge erasing moves, to at most three coordinate-parallel flat segments of effective width $W=w+4$.
By Lemmas~\ref{lem:qtrcc-corner-erasing} and \ref{lem:qtrcc-good-edge-erasing}, each of these local moves has a valid \QtRCC{} counterpart with the same support containment.
Thus the same flat reduction is valid in the qutrit model.

For each flat segment, Lemma~\ref{lem:qtrcc-confusing} and the row induction above imply that any flat leg longer than $L_{\rm flat}=9W+16$ is equivalent to one with an empty middle interval.
Such a leg disconnects the representative, because the other two coordinate legs meet it only within bounded junction regions near its ends.
Hence if the total $\ell_1$ length exceeds
\begin{equation}
3L_{\rm flat}=27W+48,
\end{equation}
at least one of the three coordinate legs disconnects the segment.
Therefore a nontrivial $X$-type string segment must satisfy
\begin{equation}
\ell\le 27W+48 .
\end{equation}

The $Z$-type proof is the reflected CSS version of the same argument.
The $X$-checks have the reflected zero, mixed, and one-qutrit corner structure of the $Z$-checks, and all coefficients are nonzero.
Corner erasing, good-edge erasing, the flat reduction, and the confusing-constraint argument therefore apply after exchanging $X$ and $Z$.

Finally, a general Pauli operator decomposes, up to phase, as
\begin{equation}
P=P_XP_Z .
\end{equation}
Because the code is CSS, $P_X$ is constrained by the $Z$-checks and $P_Z$ by the $X$-checks.
Choose the Haah corridor skeleton from the site support of the full operator $P$.
The $X$-type component is cleaned by $X$-checks using commutation with the $Z$-checks, while the $Z$-type component is cleaned by $Z$-checks using the reflected argument with the $X$-checks.
The two cleanings use the same corridor geometry and the same anchor collar.
Therefore, when a flat leg is longer than $L_{\rm flat}$, both Pauli components vanish in the same middle interval.
The product $P_XP_Z$ is then equivalent to a representative whose site support is disconnected.
Hence every nontrivial width-$w$ segment obeys Eq.~\eqref{eq:nostring-bound}.
\end{proof}

\end{document}